\documentclass[a4paper,onecolumn,floatfix,superscriptaddress,pra,showpacs,accepted=2026-03-06]{quantumarticle}
\pdfoutput=1

\usepackage{amsfonts}
\usepackage[latin1]{inputenc}
\usepackage{amsmath}
\usepackage{amsthm}
\newtheorem{lemma}{Lemma}
\usepackage{setspace}
\usepackage{braket}
\usepackage{dsfont}
\usepackage{graphicx}
\usepackage{svg}
\usepackage{float}
\usepackage{hyperref}
\hypersetup{
    colorlinks=true
}
\newtheorem{theorem}{Theorem}[section]

\begin{document}

\title{Two-particle scattering on general graphs}
\author{Luna L. Keller}
\affiliation{International Institute of Physics, Universidade Federal do Rio Grande do Norte, 59078-970, Natal, Brazil}
\author{Daniel J. Brod}
\affiliation{Instituto de F\'isica, Universidade Federal Fluminense, Niter\'oi, RJ, 24210-340, Brazil}

\begin{abstract}
Quantum walks in general graphs, or more specifically scattering on graphs, encompass enough complexity to perform universal quantum computation. Any given quantum circuit can be broken down into single- and two-qubit gates, which can then be translated into subgraphs --- gadgets --- that implement such unitaries on the logical qubits, simulated by particles traveling along a sparse graph. In this work, we start to develop a full theory of multi-particle scattering on graphs and give initial applications to build multi-particle gadgets with different properties.
\end{abstract}

\maketitle

\section{Introduction}

Continuous-time quantum walks (CTQWs) have long been used to devise quantum algorithms that outperform their classical counterparts. They were introduced for this purpose by Farhi in \cite{qcdectrees1} to devise an algorithm that solves some types of decision trees exponentially faster than any classical algorithm. In a similar fashion, Childs et al., presented in \cite{gluedtrees2} another exponential speedup, this time for searching on a structured graph. A particular case of CTQWs is scattering on graphs; already introduced in \cite{qcdectrees1}, they make use of the graph as a scattering center, with which incoming particles interact and then are transmitted in different directions. The work by Childs, Gosset and Webb \cite{multiparticleqw} shows that it is possible to simulate any $n-$qubit quantum circuit efficiently by scattering $n+1$ particles on a sparse graph, built with blocks called gadgets --- subgraphs that implement some quantum operation on the particles. Even though formally it is considered a multi-particle scattering process, the quantum gates in \cite{multiparticleqw} and other gadgets operate mostly in the regime where one particle is propagating far from any others, essentially performing a single-particle quantum walk. Interactions are required only for two-qubit controlled-phase gates, where two particles scatter off of each other on the simplest graph, a very long line, via a translation-invariant Hamiltonian. This (approximate) translation invariance implies conservation of momentum which, together with energy conservation, severely restricts the outcome of the scattering, as the momenta of the particles are individually preserved.

In this work, we solve the more general problem of two-particle scattering on general graphs. Inspired by scattering theory in continuous space, we employ the Lippmann-Schwinger formalism \cite{sakurai} to find the S-matrix for this scenario. Crucially, in contrast to previous works \cite{levinsonthm2,varbanov, multiparticleqw} that only compute the S-matrix for fixed particle momenta, our approach explicitly treats the S-matrix as an operator on the full Hilbert space, allowing us to describe momentum exchanges between the particles.

We also note that multi-particle scattering theory has broad applications, e.g.\ in condensed matter and particle physics, and its discrete-space counterpart has been much less explored, which suggests applications for our results beyond the field of quantum computing.

This work is organized as follows. In Sec.\ \ref{sec:2}, we solve the single-particle scattering problem using the Lippman-Schwinger formalism, which will also be employed for the two-particle case. In Sec. \ref{sec:3A} we review some technical details regarding scattering theory, and obtain the two-particle S-matrix in Sec.\ \ref{subsec:3B}. In Sec.\ \ref{sec:4} we describe some applications as consequences of the main result: in Sec.\ \ref{sec:4A} we study the scattering of one particle when another particle is bound to the graph, defining elastic and inelastic transmission probabilities, as well as ejection probabilities, with examples shown in Sec.\ \ref{sec:4B} for specific graphs. In Sec.\ \ref{sec:4C} we study scattering of two \textit{traveling} particles. In Sec.\ \ref{sec:4D} we define a two-particle cross section, a quantity that measures the amount of interaction during the scattering of two traveling particles. In Sec.\ \ref{sec:conclusion} we provide concluding remarks and discuss open questions.

\section{Single-particle scattering on general graphs} \label{sec:2}

Let us begin by solving the single-particle scattering problem on graphs; there are already solutions for this problem in the literature \cite{levinsonthm2}, which are ansatz-based. However, these methods do not work straightforwardly for two particles, which we will need to address formally with the well-known Lippman-Schwinger method from scattering theory. Thus, we apply the method first to a case with known solution to establish the main ideas.

The setting we consider consists of a continuous-time quantum walk on $N$ semi-infinite lines, or rails, onto which the particles can be scattered, attached to a finite graph $G$, playing the role of the scattering center, as pictured in Fig. \ref{fig:graph parametrized}. We consider, without loss of generality\footnote{If we want to consider a graph with two rails attached to one vertex, we can redefine the inner graph to include one more vertex from each rail.}, that each of the rails connects to a different vertex of $G$, let us say, vertices $\{1,2,\cdots,N\}=:\partial G$, and let us denote the $M$ internal vertices by $\{N+1,N+2,\cdots,N+M\}=:G^0$, in such a way that the set of vertices of $G$ is the disjoint union $V(G)=\partial G\cup G^0$. Using this ordering, we can write the adjacency matrix of the finite graph $G$ as
\begin{equation}
\begin{aligned}
\begin{bmatrix}
    A & B^\dagger \\
    B & D
\end{bmatrix},
\label{H matrix}
\end{aligned}
\end{equation}
where $A$ is $N\times N$, $B$ is $M\times N$ and $D$ is $M\times M$. 

\begin{figure}[h]
    \centering    \includegraphics[width=\textwidth]{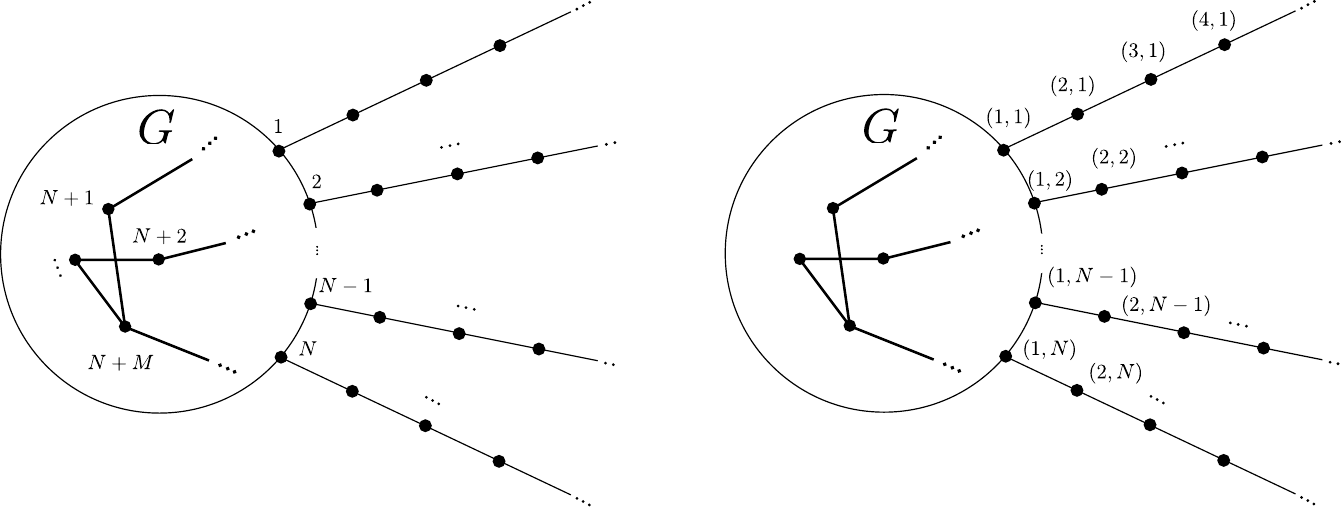}
    \caption{On the left we represent the labels for the vertices of the finite graph G. On the right we label the vertices on each rail. Note that the vertices connecting to rails have two labels, i.e., $n$ is equivalent to $(1,n)$ for $n\in\{1,2,\cdots,N\}$.}
    \label{fig:graph parametrized}
\end{figure}

We denote the position basis vectors on each vertex of $G$ by $\ket{u}$ for $u\in V(G)$. The position basis vectors on the rails are parametrized by the rail and their distance from the graph: $\{\ket{x;n}| x\in\mathbb{Z}, x\geq 1, n\in\{1,\cdots,N\} \}$, with $\ket{1;n}$ coinciding with $\ket{n}$. Denoting the set of edges by $E(G)$, the Hamiltonian for single-particle scattering on this graph is given by the adjacency matrix of the full graph:
\begin{equation}
\begin{aligned}
H=\sum_{n=1}^N \sum_{x\ge 1} \Big(\ket{x;n}\bra{x+1;n}+\ket{x+1;n}\bra{x;n}\Big) + \sum_{uv\in E(G)}\ket{u}\bra{v}=:H_{\text{free}}+H_G,
\end{aligned}
\end{equation}
where $H_G$ is the operator which, when restricted to $G$, is represented in position basis by the adjacency matrix of the graph $G$ \eqref{H matrix}, and $H_\text{free}$ is the corresponding operator on the $N$ rails. The choice of $H_G$ as the adjacency matrix of the graph, rather than a more general Hamiltonian, follows the literature on scattering quantum walks (e.g., \cite{multiparticleqw}), and it allows us to single out phenomena that are purely due to the specification of $G$ itself, since the Hamiltonian is defined in the same way at each vertex. However, the formalism described here would be applicable for more general Hamiltonians as well (e.g., with different weights for each edge).

Let us describe the eigenvectors of the free Hamiltonian $H_{\text{free}}$. Since it is a direct sum of operators acting on orthogonal subspaces, namely, each subspace corresponding to a single rail, and the subspace corresponding to $G$, its eigenvectors can be chosen inside each subspace. So we choose the following basis for scattering states\footnote{These are normalized scattering states on a semi-infinite line, with real and positive coefficient on the incoming wave.}:
\begin{equation}
\begin{aligned}
\braket{x;m|p^n}&=\dfrac{\delta_{mn}}{\sqrt{2\pi}}(e^{-ipx}-e^{ipx})=\dfrac{\delta_{mn}}{\sqrt{2\pi}}(z^{-x}-z^{x}) &\text{ for } x\ge 1 \text{ and } m=1,\cdots,N,
\end{aligned}
\end{equation}
where $z:=e^{ip}$, with eigenvalues $E_p=2\cos{p}=z+z^{-1}$. Let us define the amplitudes of $\ket{p^n}$ on the graph $G$ to be zero. Note that making $p\xrightarrow{}-p$ only changes the state by a global phase. The authors in \cite{levinsonthm2} show we only need to consider half of the interval $[-\pi,\pi)$ to form a complete basis (alongside possible bound states). We follow their convention and consider $p\in[-\pi,0)$ for incoming momenta. These states, together with the vectors $\ket{u}$ for $u\in G_0$ complete an eigenbasis for $H_{\text{free}}$. We will associate $k\in[0,\pi)$ with \emph{outgoing} momenta; although this is an arbitrary choice, it will make the results compatible with \cite{levinsonthm2}.

Given incoming and outgoing momenta, $p\in (-\pi,0)$ and $k\in (0,\pi)$ respectively, we are interested in computing the S-matrix, defined by \cite{sakurai,taylor}:
\begin{equation}
\begin{aligned}
S_{k,p}^{m,n}=\braket{{k^m}^-|{p^n}^+},
\end{aligned}
\end{equation}
where $\ket{{p^n}^\pm}$ are scattering states, i.e., eigenstates of the interacting Hamiltonian which converge, in the asymptotic past ($+$) or future ($-$), to the free-particle states  $\ket{p^n}$. In other words, $S_{k,p}^{m,n}$ is the probability amplitude of measuring the state $\ket{k^m}$ in the distant future (an outgoing particle with momentum $k$ on the rail $m$) given that we prepared the state $\ket{p^n}$ in the distant past (an incoming particle with momentum $p$ on the rail $n$). A more precise formulation can be made by considering wavepackets; see \cite{taylor}. These states can be found via Lippmann-Schwinger equation; it relates the scattering states in the presence of an interaction to the free-particle scattering states:
\begin{equation}
\begin{aligned}
\ket{{p^n}^\pm}=\ket{{p^n}} + (E-H^0\pm i\epsilon)^{-1}V\ket{{p^n}^\pm},
\label{singlelipp}
\end{aligned}
\end{equation}
where $V$ is the interaction considered in addition to a free Hamiltonian $H_0$ and we take the limit $\epsilon\rightarrow 0^+$ at the end of the computation. For the single-particle scattering we consider $H_0\equiv H_\text{free}$ and $V\equiv H_G$. Based on this equation, it is possible to arrive at a simpler expression for the S-matrix:
\begin{equation} \label{eq:singleparticleS}
\begin{aligned}
S_{k,p}^{m,n}=-\delta_{mn}\delta(k+p)-2\pi i\delta(E_k-E_p)\bra{k^m}V\ket{{p^n}^+}.
\end{aligned}
\end{equation}
Since $p\in(-\pi,0)$ and $k\in(0,\pi)$, we can write $\delta(E_k-E_p)=\delta(2\cos{k}-2\cos{p})=\dfrac{\delta(k+p)}{|2\sin{p}|}=\dfrac{\delta(k+p)}{i(z-z^{-1})}$, which is a manifestation of the fact that single-particle energy conservation implies conservation of momentum as well. As a consequence, the S-matrix has a block diagonal form with respect to the momentum variables:
\begin{equation}
\begin{aligned}
S_{k,p}^{m,n}=\delta(k+p)\left(-\delta_{mn}+\dfrac{2\pi}{z-z^{-1}}\bra{p^m}V\ket{{p^n}^+}\right)=:\delta(k+p)\mathbb{S}_{mn}(z),
\end{aligned}
\end{equation}
where we defined, for each fixed $z$, the $N\times N$ S-matrix  $\mathbb{S}(z)$. The amplitudes of $\ket{p^{n+}}$ on the rails can now be written as
\begin{equation}
\begin{aligned}
\braket{x;m|p^{n+}}&=\dfrac{1}{\sqrt{2\pi}}(\delta_{mn}e^{-ipx}+\mathbb{S}_{mn}e^{ipx}) &\text{ for } x\ge 1 \text{ and } m=1,\cdots,N.
\label{singlefree}
\end{aligned}
\end{equation}

In appendix \ref{apx:A} we prove the above equation, compute the S-matrix and the amplitudes of the scattering states for inner vertices $\braket{u|p^{n+}}=:\frac{1}{\sqrt{2\pi}}\Psi_{un}$. Defining $Q(z):=\mathds{1}-zA-zB^\dagger(z+z^{-1}-D)^{-1}B$, the full solution to single-particle scattering is summarized in the following matrices:
\begin{equation}
\begin{aligned}
\mathbb{S}(z)&=-Q(z^{-1})Q(z)^{-1},\\
\Psi(z)&=(z+z^{-1}-D)^{-1}B(z^{-1}\mathds{1}+z\mathbb{S}(z)).
\end{aligned}
\end{equation}
The matrices $\Psi$ and $\mathbb{S}$ give a complete description of the single-particle scattering states. However, these states do not form a basis for the full Hilbert space; together with the non-normalized states defined above, there are also normalizable states, called bound states (which can be further divided into confined and evanescent bound states, with slightly different physical properties). They will also play a role in two-particle scattering, and more details are found in Appendix \ref{apx:B}.

\section{Main result} \label{sec3}

In this section we obtain our main result: a description of the S-matrix of two interacting particles scattering on a general graph. However, before that, we will take a detour on formal scattering theory to find the correct domain of the S-matrix. It turns out that there can be no scattering of two traveling particles resulting in two (single-particle) bound states to the graph, or the time-reversed process. Intuitively, two particles initialized in \textit{single-particle} bound states will interact and will not both remain bound to the graph asymptotically, unless their initial state coincides with a \textit{two-particle} bound state. If this intuition is convincing enough, feel free to skip to Sec.\ \ref{subsec:3B}.

\subsection{A little detour on formal scattering theory}
\label{sec:3A}

A quantum scattering problem involves a particle that can move freely in some region and a scattering center with which the particle interacts. The free particle is described by some Hamiltonian $H_0$, which may be the Laplacian for particles in $\mathbb{R}^3$ or, in our case, the adjacency matrix of the underlying graph for particles restricted to vertices of such graph. The scattering center may be  a heavy particle which produces a localized potential well which may redirect the path of the scattered particle, connecting two different free asymptotes or free states; let us call the interacting Hamiltonian $H$. In the previous section, the role of the scattering center was fulfilled by the scattering graph $G$. We denote the free Hamiltonian by $H_0$ and the scattering center is represented by another operator $V$, both acting on a Hilbert space $\mathcal{H}$.

Let $U(t)$ and $U_0(t)$ represent the unitary evolution operator for $H$ and $H_0$ respectively, and let $\ket{\psi}$ describe the state of the particle at some specific time during the scattering. Let $U(t)\ket{\psi}$ describe a scattering experiment: we expect that in the distant past, the particle was far from the scattering center, and so behaved as a free particle, in the following sense
\begin{equation}
\begin{aligned}
\lim_{t\rightarrow -\infty}||U(t)\ket{\psi}-U_0(t)\ket{\psi_{\text{in}}}||= 0
\label{inasymptote}
\end{aligned}
\end{equation}
for some in-asymptote $\ket{\psi_{\text{in}}}$. Analogously, in the distant future we expect that the particle is also asymptotically free, meaning
\begin{equation}
\begin{aligned}
\lim_{t\rightarrow \infty}||U(t)\ket{\psi}-U_0(t)\ket{\psi_{\text{out}}}||= 0
\label{outasymptote}
\end{aligned}
\end{equation}
for some out-asymptote $\ket{\psi_{\text{out}}}$. This description raises several questions about the consistency of this problem, which in general are very hard to prove. Let us look at the first question.

\textbf{The asymptotic condition}: Given a state in $\mathcal{H}$, does it represent an in-asymptote of some actual orbit? In other words, given $\ket{\psi_{\text{in}}}\in\mathcal{H}$, is there a state $\ket{\psi}$ that satisfies \eqref{inasymptote}? And similarly for out-asymptotes.

For arbitrary potentials we do not know the answer; but it can be proven that some types of potentials satisfy the asymptotic condition. For example, in $\mathbb{R}^3$, this condition holds for square integrable central potentials \cite{cook}, or even central potentials which are square integrable near the origin and fall faster than $r^{-1}$ at infinity \cite{hack}. In our case, the asymptotic condition is \textit{not} generally satisfied; this is a consequence of the existence of bound states of the \textit{free} Hamiltonian.

We know from spectral theory that the spectrum of a self-adjoint operator can be decomposed into three parts: a discrete, or pure point spectrum, associated with bound states; an absolutely continuous spectrum, associated with scattering states; and a residual, or singular spectrum, that can be associated with the particle undergoing a \textit{quantum walk} \cite{Pearson}. With those considerations, the generalized wave operators, or M\o ller operators, are defined by the following limits, if they exist: \footnote{The odd sign convention is connected to the Lippman-Schwinger formalism, which we will resort to later on.}
\begin{equation}
\begin{aligned}
\Omega^{\pm}:=\lim_{t\rightarrow \mp\infty} U(t)^\dagger U_0(t)P_{\text{ac}}(H_0),
\end{aligned}
\end{equation}
where $P_{\text{ac}}(H_0)$ is the projection on the absolutely continuous subspace of $H_0$. Being the limit of unitary operators composed with a projection, the wave operators are partial isometries with initial subspace $P_{\text{ac}}(H_0)\mathcal{H}=:\mathcal{H}_{\text{ac}}$ and some final subspaces $\mathcal{H}_{\pm}:=\text{Ran}(\Omega_\pm)$; partial isometries are characterized by the equations $\Omega\Omega^\dagger=P_{\text{Ran}(\Omega)}$ and $\Omega^\dagger\Omega=P$ where $P$ is the projection on its support.

Returning to our scattering experiment, let us rephrase it as follows: we prepare a particle in some in-asymptote, evolving as $e^{-iH_0 t}\ket{\phi}$ for $t\rightarrow -\infty$, during the scattering it evolves as $e^{-iH t}\Omega^+\ket{\phi}$ and we ask for the probability amplitude of finding the out-asymptote $e^{-iH_0 t}\ket{\psi}$ for $t\rightarrow \infty$. The state that evolved to the out-asymptote is described as $e^{-iH t}\Omega^-\ket{\psi}$ during the scattering. The amplitude of the process $\phi\rightarrow\psi$ is then given by the inner product
\begin{equation}
\begin{aligned}
\bra{\psi}(\Omega^-)^\dagger\Omega^+\ket{\phi}=:\bra{\psi}S\ket{\phi},
\end{aligned}
\end{equation}
where we have defined the S-matrix $S:=(\Omega^-)^\dagger\Omega^+$ acting on the Hilbert space $\mathcal{H}_{\text{ac}}$. Note that $S$ is not automatically unitary:
\begin{equation}
\begin{aligned}
S^\dagger S&=(\Omega^+)^\dagger\Omega^-(\Omega^-)^\dagger\Omega^+=(\Omega^+)^\dagger P_{\mathcal{H}_-}\Omega^+\\
&=(\Omega^+)^\dagger \Omega^+= P_{\text{ac}}(H_0),
\end{aligned}
\end{equation}
where the second line follows if, and only if $\mathcal{H}_+=\mathcal{H}_-$. The same argument holds for $SS^\dagger$. We will now study those spaces with more detail. As discussed before, we can think of $\mathcal{H}_+$ ($\mathcal{H}_-$) as the subspace of states that have an in(out)-asymptote, or scattering states. We expect that these states together with the bound states of the interacting system form an orthonormal basis for $\mathcal{H}$. This is always true:

\begin{theorem}\textbf{Orthogonality theorem:} Let $\mathcal{B}$ be the subspace of bound states of $H$. Then $\mathcal{H}_\pm$ are orthogonal to $\mathcal{B}$.
\end{theorem}

The next condition is more involved, and is not necessarily guaranteed for any given physical system.

\begin{theorem}
\textbf{Asymptotic completeness:} Suppose that $\Omega^\pm$ exist. They are said to be \textbf{complete} if 
\begin{equation}
\begin{aligned}
\text{Ran}(\Omega^+)=\text{Ran}(\Omega^-)=\text{Ran}(P_{\text{ac}}(H)).
\end{aligned}
\end{equation}

In addition, if $\text{Ran}(P_{\text{ac}}(H))=\mathcal{B}^\perp$, we say that they are \textbf{asymptotically complete}.\footnote{For more details on the distinction, see for example \cite{Pearson}.}
\end{theorem}

In the case of scattering on graphs, considering the interaction to be localized, e.g., the particles interact only on the graph $G$ and not on the rails, completeness is guaranteed by the following theorem (see \cite{reedsimon3}):

\begin{theorem}
\textbf{(Kato-Rosenblum 1957):} Given $H_0,H$ self-adjoint operators with $H-H_0$ trace class, then the wave operators exist and are complete.    
\end{theorem}

We now describe the domain of the S-matrix; for single-particle scattering the absolutely continuous subspace of $H_0$ is spanned by the free states $\ket{p^n}$, and the pure point subspace, or 'bound states' is spanned by $\ket{u}$, $u\in G_0$.

The situation for two-particle scattering is a bit more complicated. As we define in the next section, the free Hamiltonian for each particle is the full single-particle Hamiltonian described in the last section---namely, both the rails and the scattering graph $G$---, and the potential $V$ is a Bose-Hubbard type interaction between the particles. So the free states are spanned by tensor products of two scattering states on the graph $\ket{p^{n+}}$ and/or single particle bound states $\ket{\chi}$, excluding combinations of two bound states \footnote{Since they lie in the pure point spectrum of the Hamiltonian.}, in such a way that a basis may be chosen to consist of the states $\ket{p_1^{n+}}\otimes\ket{p_2^{n+}}$, $\ket{p_1^{n+}}\otimes\ket{\chi}$ and $\ket{\chi}\otimes\ket{p_2^{n+}}$, for $p_1,p_2\in(-\pi,0)$, $n=1,\cdots,N$ and $\chi$ in some basis for the bound states. Henceforth, we will omit the $\otimes$ in the tensor product notation for brevity. In the following section, in order to compute the S-matrix, we shall again use the Lippmann-Schwinger formalism, applied to two-particle scattering. For an informal derivation of this formalism applied to discrete space systems, refer to \cite{ryndyk} chapter 3; for more technical remarks, refer to \cite{reedsimon3} chapter XI.6 and the references therein.

\subsection{Two-particle S-matrix}
\label{subsec:3B}
Before we consider the scattering problem, let us comment on the notation used for bosons and distinguishable particles. We often denote two-particle states with double label kets, $\ket{\zeta_1\zeta_2}$, without making a particular choice of symmetrization. Whenever we refer to distinguishable particles this should be interpreted as $\ket{\zeta_1\zeta_2}=\ket{\zeta_1}\ket{\zeta_2}$, whereas for bosons and fermions it should be interpreted as $\ket{\zeta_1\zeta_2}=(\ket{\zeta_1}\ket{\zeta_2}\pm\ket{\zeta_2}\ket{\zeta_1})/\sqrt{2}$, as appropriate. Whenever particle statistics is unspecified, such as in the lead-up to our main result, the statements should be interpreted as holding for all particle types. When the distinction becomes relevant, such as for numerical results, the particle statistics will be stated explicitly.

We now consider the scattering of two particles on an arbitrary graph $G$. As before,  the full graph where the scattering takes place is a finite graph $G$ appended with $N$ rails, namely, semi-infinite lines. In addition, we will assume that the particles interact via a Bose-Hubbard interaction\footnote{In appendix \ref{apx:E} we consider more general potentials.} restricted to the internal vertices of the graph. Denoting $\ket{a}\bra{b}_1\equiv\ket{a}\bra{b}\otimes\mathds{1}$ and $\ket{a}\bra{b}_2\equiv\mathds{1}\otimes\ket{a}\bra{b}$, the full Hamiltonian reads as follows:
\begin{equation}
\begin{aligned}
H=H_{\text{rails}}+H_G+V,
\end{aligned}
\end{equation}
where
\begin{equation}
\begin{aligned}
H_{\text{rails}}&=\sum_{d=1}^2\sum_{n=1}^N \sum_{x\ge 1} \Big(\ket{x;n}\bra{x+1;n}_d+\ket{x+1;n}\bra{x;n}_d\Big)\\
H_G&=\sum_{d=1}^2\sum_{uv\in E(G)} \ket{u}\bra{v}_d\\
V&=U\sum_{w\in G^0}\ket{ww}\bra{ww},
\end{aligned}
\end{equation}
in other words, the particles interact with each other when they are on the internal vertices of the graph and propagate freely on the semi-infinite lines. We denote $H_0:=H_{\text{rails}}+H_G$. In the absence of the interaction, the particles still scatter on the graph, and the interaction-free states are the ones described by Eq. \eqref{singlefree}. If we want to consider incoming particles of momenta $p_1$ and $p_2$ in rails $n_1$ and $n_2$ respectively in the distant past, or if we want to consider an incoming particle of momentum $p_1$ in rail $n_1$ with the other particle in some bound state $\chi$, we use the following Lippman-Schwinger equation:
\begin{equation}
\begin{aligned}
\ket{{\xi_1\xi_2}^+}&=\ket{\xi_1\xi_2} + (E-H_0+ i\epsilon)^{-1}V\ket{{\xi_1\xi_2}^+},
\end{aligned}
\end{equation}
where $\xi_1,\xi_2\in\{\chi,{p_1^{n_1+}},{p_2^{n_2+}}\}$ \footnote{The full notation for the eigenstate corresponding to two interacting particles incoming on the graph would be $\ket{{p_1^{n_1+} p_2^{n_2+}}^+}$, which we will try to avoid writing.}. Note that $\ket{p^{n+}}$ stands for the single particle scattering state on the graph, described in Section \ref{sec:2}. Analogously, if we want to consider outgoing particles of momenta $k_1$ and $k_2$ at rails $m_1$ and $m_2$ in the distant future, or if we want to consider an outgoing particle of momentum $k_1$ in rail $m_1$ with the other particle in some bound state $\ket{\chi}$, we use the following Lippmann-Schwinger equation:
\begin{equation}
\begin{aligned}
\ket{{\zeta_1\zeta_2}^-}=\ket{\zeta_1\zeta_2} + (E-H_0- i\epsilon)^{-1}V\ket{{\zeta_1\zeta_2}^-},
\end{aligned}
\end{equation}
where $\zeta_1,\zeta_2\in\{\chi,{k_1^{m_1-}},{k_2^{m_2-}}\}$. As before, we consider ingoing momenta $p\in (-\pi,0)$ and outgoing momenta $k\in (0,\pi)$.

We are interested in computing the S-matrix, defined by:
\begin{equation}
\begin{aligned}
S_{\zeta_1\zeta_2;\xi_1\xi_2}=\braket{{\zeta_1\zeta_2}^-|{\xi_1\xi_2}^+},
\end{aligned}
\end{equation}
and note that we are using different bases for input and output asymptotes. This is similar to the setting described in \cite{inputoutput}. As in the single-particle situation, we can arrive at the simpler form:
\begin{equation}
\begin{aligned}
&S_{\zeta_1\zeta_2;\xi_1\xi_2}=\braket{\zeta_1\zeta_2|\xi_1\xi_2}-2\pi i\delta(E_{\zeta_1\zeta_2}-E_{\xi_1\xi_2})\bra{\zeta_1\zeta_2}V\ket{{\xi_1\xi_2}^+},
\end{aligned}
\end{equation}
where $E_{\zeta_1\zeta_2}$ ($E_{\xi_1\xi_2}$) is the energy of the state $\ket{\zeta_1\zeta_2}$ ($\ket{\xi_1\xi_2}$). For example, choosing the $\xi$'s and $\zeta$'s as scattering states for distinguishable particles, the above equation is written as
\begin{equation}
\begin{aligned}
&S_{k_1^{m_1} k_2^{m_2}; p_1^{n_1} p_2^{n_2}}=\braket{{k_1^{m_1}}^-|{p_1^{n_1}}^+}\braket{{k_2^{m_2}}^-|{p_2^{n_2}}^+}-2\pi i\delta_E\bra{{k_1^{m_1}}^-}\bra{{k_2^{m_2}}^-}V\ket{{(p_1^{n_1+} p_2^{n_2+}})^+}\\
&\hspace{0.5cm}=\mathbb{S}_{m_1n_1}(e^{ip_1})\mathbb{S}_{m_2n_2}(e^{ip_2})\delta(k_1+p_1)\delta(k_2+p_2)-2\pi i\delta_E\bra{{k_1^{m_1}}^-}\bra{{k_2^{m_2}}^-}V\ket{({p_1^{n_1+} p_2^{n_2+}})^+},
\end{aligned}
\end{equation}
where $\delta_E=\delta(E_{k_1k_2}-E_{p_1p_2})$.

Note that, in contrast to the single-particle case, here the individual momentum of each particle is \textit{not} necessarily conserved; after the scattering, the particles may exhibit spectral entanglement. In Appendix \ref{apx:E} we compute the analytical expression for the S-matrix:
\begin{equation}
\begin{aligned}
S_{\zeta_1\zeta_2;\xi_1\xi_2}=\braket{\zeta_1\zeta_2|\xi_1\xi_2}
-2\pi i \delta(E_{\zeta_1\zeta_2}-E_{\xi_1\xi_2})\sum_{u,v\in G^0}\braket{\zeta_1\zeta_2|uu}\left(\frac{1}{U}\mathds{1}-J\right)^{-1}_{uv}\braket{vv|\xi_1\xi_2},
\end{aligned}
\end{equation}
where $J$ is a $N\times N$ matrix defined in Eq.\ (\ref{Juv}). Note that the first term on the RHS are matrix elements of the identity operator; let us write $S=:\mathds{1}+R$, so that the above equation implies that
\begin{equation}
\begin{aligned}
R_{\zeta_1\zeta_2;\xi_1\xi_2}&=
-2\pi i \delta(E_{\zeta_1\zeta_2}-E_{\xi_1\xi_2}) \sum_{u,v\in G^0}\braket{\zeta_1\zeta_2|uu}\left(\frac{1}{U}\mathds{1}-J\right)^{-1}_{uv}\braket{vv|\xi_1\xi_2}\\
&=:\delta(E_{\zeta_1\zeta_2}-E_{\xi_1\xi_2})\mathcal{R}_{\zeta_1\zeta_2;\xi_1\xi_2}.
\end{aligned}
\end{equation}

If we impose particle statistics, we can further simplify this expression in terms of single particle amplitudes:
\begin{equation}
\begin{aligned}
R_{\zeta_1\zeta_2;\xi_1\xi_2}=
-2\pi i \delta(E_{\zeta_1\zeta_2}-E_{\xi_1\xi_2}) \sum_{u,v\in G^0}b^2 \braket{\zeta_1|u}\braket{\zeta_2|u}\left(\frac{1}{U}\mathds{1}-J\right)^{-1}_{uv}\braket{v|\xi_1}\braket{v|\xi_2},
\end{aligned}
\end{equation}
where $b^2=1$ or $2$, depending whether we consider distinguishable particles or bosons. Note that this implies that the interaction for bosons is double that of distinguishable particles, $\mathcal{R^\text{dist}}=2\mathcal{R^\text{bosons}}$; this may be seen as a consequence of boson bunching. These equations are the main result of this work; we have written an explicit formula for the two-particle S-matrix, depending on amplitudes of the free and bound single-particle states, and on the inverse of a $M\times M$ matrix whose coefficients are related to solutions to a quadratic eigenvalue equation --- this is where the properties of $\gamma(z)$ shown in Appendix \ref{apx:C} are useful. We emphasize that this formula is exact for any interaction strength $U$, it does not rely on a perturbative approach such as the Born approximation. In fact, throughout this work we often work in the limit of very large $U$ that is out of reach of the perturbative approach. We also point out that our formalism does also apply to more general potentials within the graph (with finite support) rather than just single-site interactions. We present only the Bose-Hubbard case here for simplicity, but the more general case can be found in Appendix \ref{apx:E}. 

In the next section, we apply this result to study features of two-particle graph scattering and to analyze properties of some graph families.

\section{Applications}
\label{sec:4}

We now apply the formalism developed in the previous Section to describe scattering processes in graphs in more detail. In Sec.\ \ref{sec:4A} we study how the scattering of one particle is modified by the presence of a bound particle in the graph. We classify different outcomes, such as elastic and inelastic scattering and ejection of the bound particle, and then apply these results to specific situations in Sec.\ \ref{sec:4B}. In Sec.\ \ref{sec:4C} we study the scattering of two traveling particles, as opposed to having one bound to the graph, looking at some examples where we expect that the particles may conserve their individual momenta. In Sec.\ \ref{sec:4D} we define a two-particle cross section, which is a measure of the interaction between the particles during scattering.

In the first two subsections we consider scattering states in energy eigenbasis; with respect to momentum eigenbasis they have an extra normalizing factor $\ket{E^{n\pm}}=\frac{1}{\sqrt[4]{4-E^2}}\ket{p^{n\pm}}$. This is reminiscent of the orthonormalization condition, $\braket{p|p'}=\frac{1}{2\pi}\sum_{n=-\infty}^\infty e^{i(p'-p)n}=\delta(p-p')$, on the line lattice:
\begin{equation}
\begin{aligned}
\braket{E|E'}=\delta(E-E')=\delta(2\cos(p)-2\cos(p'))=\frac{1}{|2\sin{p}|}\delta(p-p')=\frac{\braket{p|p'}}{\sqrt{4-E^2}}.
\end{aligned}
\end{equation}

Note that this is just a consequence of the dispersion relation $E=2\cos(p)$, and passes on to more general graphs.

\subsection{Single-particle scattering in the presence of a bound state}
\label{sec:4A}

For single-particle scattering, as in Sec. \ref{sec:2}, the matrix coefficients $t^{mn}(E):=\mathbb{S}^{mn}(e^{ip})$, called transmission\footnote{Or reflection coefficient if $n=m$.} coefficients, represent the amplitude of the process where the particle is incoming at rail $n$ with energy $E\in(-2,2)$ and scatters to rail $m$, satisfying energy conservation. Let us check that $|t^{mn}(E)|^2$ is the actual probability of finding the particle at rail $m$ after scattering; consider a wavepacket $\ket{\psi(E^n)}$ peaked around energy $E$ incoming at rail $n$. That probability can be calculated as follows:
\begin{equation}
\begin{aligned}
||\mathbb{P}^{m}S\ket{\psi(E^n)}||^2,
\end{aligned}
\end{equation}
where $\mathbb{P}_{m}:=\int_{-2}^2 \ket{\mathcal{E'}^m}\bra{\mathcal{E'}^m}d\mathcal{E'}$ is the projector onto rail $m$. This quantity can be computed directly; consider the state $\ket{\psi(E^n)}:=\frac{1}{\sqrt{2\epsilon}}\int_{E-\epsilon}^{E+\epsilon}\ket{\mathcal{E}^n}d\mathcal{E}$. Then, we can write
\begin{equation}
\begin{aligned}
\bra{\mathcal{E'}^m}S\ket{\psi(E^n)}&=\frac{1}{\sqrt{2\epsilon}}\int_{E-\epsilon}^{E+\epsilon}\bra{\mathcal{E'}^m}S\ket{\mathcal{E}^n}d\mathcal{E}\\
&=\frac{1}{\sqrt{2\epsilon}}\int_{E-\epsilon}^{E+\epsilon}t^{mn}(\mathcal{E})\delta(\mathcal{E}'-\mathcal{E})d\mathcal{E}=\frac{1}{\sqrt{2\epsilon}}t^{mn}(\mathcal{E}')\mathds{1}_{|\mathcal{E}'-\mathcal{E}|<\epsilon}\\
||\mathbb{P}^{m}S\ket{\psi(E^n)}||^2
&=\frac{1}{2\epsilon}\int_{-2}^{2}|t^{mn}(\mathcal{E}')|^2\mathds{1}_{|\mathcal{E}'-E|<\epsilon} d\mathcal{E}'
\longrightarrow|t^{mn}(E)|^2,
\end{aligned}
\end{equation}
where $\mathds{1}_{\mathcal{C}}$ is the indicator function, assuming value $1$ if the condition $\mathcal{C}$ is satisfied and $0$ otherwise, and in the final step we let $\epsilon\rightarrow0^+$. 

Let us rephrase these steps for two-particle scattering. Consider a single particle undergoing scattering on the graph with the presence of a second particle in some bound state. The probability of the first particle scattering from rail $n$ to rail $m$ is:
\begin{equation}
\begin{aligned}
||\mathbb{P}_{(2)}^{m}S\ket{\psi(E^n);\chi}||^2 
\end{aligned}
\end{equation}
where $\ket{\psi(E^n);\chi}$ is thought of as the first particle with energy peaked around $E$ incoming at rail $n$ and the other particle in the bound state $\ket{\chi}$, and $\mathbb{P}_{(2)}^{m}$ is the projector that measures if the first particle is found traveling in rail $m$ at $t\rightarrow\infty$, ignoring the state of the other particle. For distinguishable particles, the projector reads\footnote{For bosons, we just need to restrict the integral to $\mathcal{E}<\mathcal{E'}$ if $m'=m$, to avoid double counting.}
\begin{equation}
\begin{aligned}
\mathbb{P}^{m}_{(2)}&=\int_{-2}^2\sum_{\chi'} \ket{{\mathcal{E}}^{m-}\chi'}\bra{{\mathcal{E}}^{m-}\chi'}d\mathcal{E} + \int_{-2}^2\int_{-2}^2\sum_{m'} \ket{{\mathcal{E}^m}^-{\mathcal{E}'\,^{m'-}}}\bra{{\mathcal{E}^m}^-{\mathcal{E}'\,^{m'-}}}d\mathcal{E}d\mathcal{E}'\\
&:=\sum_{\chi'}\mathbb{P}^{m}_{\chi'}+\sum_{m'}\mathbb{P}^{m}_{m'},
\end{aligned}
\end{equation}
where we split the projector further into terms of the type $\mathbb{P}^{m}_{a}$, where $a$ is either the bound state the second particle is found, or the rail it scatters into.

Using the condition of narrow wavepackets, we can show that\footnote{Again, for bosons, we need to restrict the integral to $\mathcal{E}<\mathcal{E}_{\text{cons}}$ if $m'=m$, to avoid double counting.}
\begin{equation}
\begin{aligned}
||\mathbb{P}_{(2)}^{m}S\ket{\psi(E^n);\chi}||^2&= \sum_{\chi'}||\mathbb{P}_{\chi'}^{m}S\ket{\psi(E^n);\chi}||^2+\sum_{m'}||\mathbb{P}_{m'}^{m}S\ket{\psi(E^n);\chi}||^2\\
&= |t^{mn}(E)+\mathcal{R}_{E^{m-}\chi;{E^{n+}}\chi}|^2 + \sum_{\chi'\neq\chi}|\mathcal{R}_{E_{cons}^{m-}\chi';{E^{n+}}\chi}|^2\\ &+ \int_{-2}^2\sum_{m'}|\mathcal{R}_{\mathcal{E}^{m-}\mathcal{E}_{\text{cons}}^{m'-};E^{n+}\chi}|^2\mathds{1}_{|\mathcal{E}_{\text{cons}}|\leq 2} d\mathcal{E}\\
&=:|t^{mn}_{\chi\rightarrow\chi}(E)|^2+ \sum_{\chi'\neq\chi}|t^{mn}_{\chi\rightarrow\chi'}(E)|^2+p^{mn}_{ej}(E,\chi)
\end{aligned}
\end{equation}
where $E_{cons}=E+E_\chi-E_{\chi'}$ and $\mathcal{E}_{\text{cons}}=E+E_\chi-\mathcal{E}$. Note that the probabilities of scattering are modified by the presence of a particle bound to the graph through every possible process that is available; the first term represents \textit{elastic transmission}, that is, $t^{mn}_{\chi\rightarrow\chi}(E)$ is the amplitude for the process where the incoming particle scatters on the graph and on the second particle, with no energy exchange between them. The next family of processes we call \textit{inelastic transmission}, where $t^{mn}_{\chi\rightarrow\chi'}(E)$ is the amplitude of the scattering particle exciting the bound particle to a different bound state, leaving the graph with energy $E_{cons}$, respecting energy conservation. The last process is the ejection of the bound state, where both particles end up on free scattering states, and $p^{mn}_{ej}(E,\chi)$ is the probability of finding the first particle on rail $m$ and the second particle on \textit{any} rail after the scattering, and $\mathcal{R}_{\mathcal{E}^{m-}\mathcal{E}_{\text{cons}}^{m'-};E^{n+}\chi}$ is the corresponding amplitude; again, by conservation of energy, the energy of the second particle after the scattering is $\mathcal{E}_{\text{cons}}$.

For distinguishable particles, there is a different process, where the first particle ends trapped on the graph while the second particle is ejected; for indistinguishable particles this is already included in elastic/inelastic transmission. For completeness, we write the corresponding projector and probabilities associated with these processes:
\begin{equation}
\begin{aligned}
\mathbb{P}^{\text{bound;m}}_{(2)}&=\int_{-2}^2\sum_{\chi'} \ket{{\chi'\mathcal{E}}^{m-}}\bra{{\chi'\mathcal{E}}^{m-}}d\mathcal{E} \\
||\mathbb{P}_{(2)}^{\text{bound;m}}S\ket{\psi(E^n);\chi}||^2&=\sum_{\chi'}|\mathcal{R}_{\chi'E_{cons}^{m-};{E^{n+}}\chi}|^2
\end{aligned}
\end{equation}

We note that, following the discussion of Section \ref{sec:3A}, and for the two-particle interaction considered, there is no process where both particles remain on bound states asymptotically.




\subsection{Analysis of scattering with a bound particle} 
\label{sec:4B}
Now we apply the classification discussed above to some specific graphs; for a list of the graphs presented in this Section and their properties, see \ref{apx:D}. We study graphs with only two rails, so all the scattering matrices for a single particle $2\times2$, and we interpret elements $\mathbb{S}_{11}$ and $\mathbb{S}_{21}$ as reflection coefficients $r$ and transmission coefficients $t$, respectively. Also, we consider the Bose-Hubbard interaction and we set the coupling constant $U=\infty$ henceforth.

We first consider the graph \texttt{AC(4)}, shown in Fig.\ \ref{fig:AC2(4)} a). In Fig.\ \ref{fig:AC2(4)} we plot $|r|^2$ and $|t|^2$ for single-particle scattering. Note, for example, that we have perfect transmission at $E=0$ and perfect reflection at $E=\pm\sqrt{2}$.

\begin{figure}[h]
    \centering
    \includegraphics[width=0.9\textwidth]{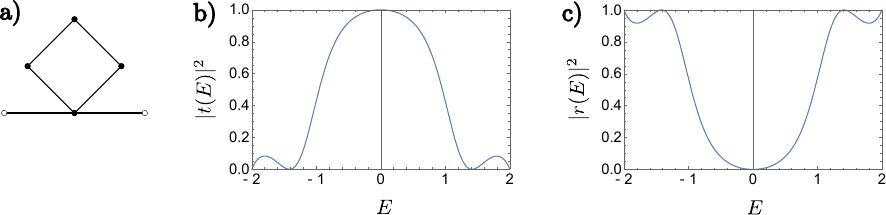}
    \caption[\texttt{AC(4)}]{From left to right: The graph \texttt{AC(4)}, the single particle transmission coefficient $|t|^2$ and the single particle reflection coefficient $|r|^2$ as a function of the energy of the incident particle.}
    \label{fig:AC2(4)}
\end{figure}

We now analyze two-particle scattering on this graph; first, we note that for this graph there are two unconfined bound states $\ket{\chi_\pm}$ with energies near $E_{ev}=\pm 2.38$, and one confined bound state $\ket{\chi_c}$ with energy $E_c=0$. As we discussed in the last section, we will consider one particle at some bound state on the graph and the other incoming particle with energy $E\in (-2,2)$, and will study the non-linear scattering coefficients. Let us begin with the elastic reflection/transmission amplitudes.
 
In Fig.\  \ref{fig:appendcyc4 r and t evanescent} we plot the elastic reflection probabilities for distinguishable particles and bosons, $|r_{\chi\rightarrow\chi}^D|^2$ and $|r_{\chi\rightarrow\chi}^B|^2$ respectively, and the elastic transmission probabilities $|t_{\chi\rightarrow\chi}^D|^2$ and $|t_{\chi\rightarrow\chi}^B|^2$, in the presence of each evanescent bound state:

\begin{figure}[H]
    \centering
    \includegraphics[width=0.9\textwidth]{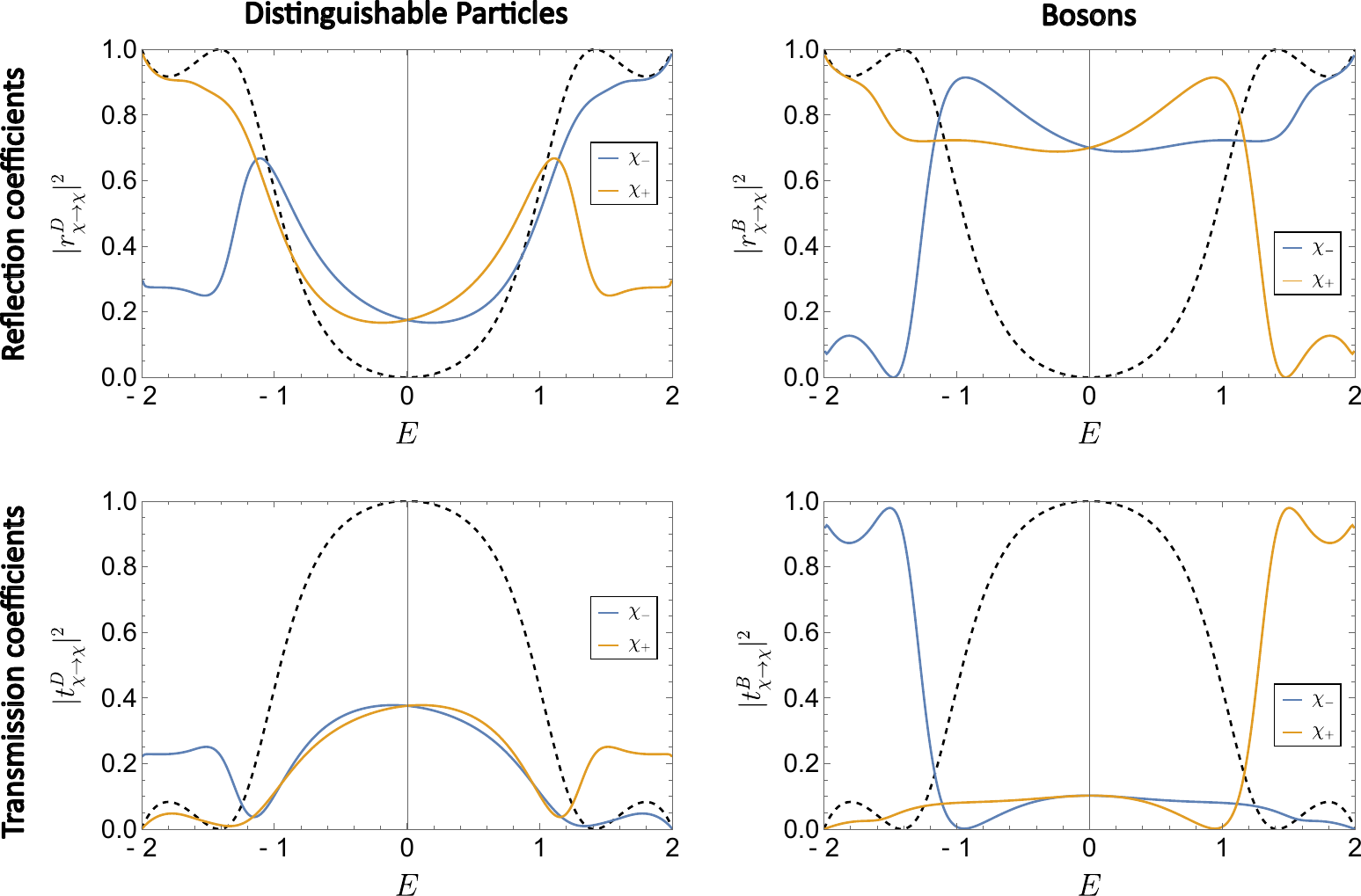}
    \caption[\texttt{AC(4)}: elastic transmission and reflection probabilities with evanescent bound state]{Distinguishable particle elastic reflection and transmission probabilities, $|r_{\chi\rightarrow\chi}^D|^2$ and $|t_{\chi\rightarrow\chi}^D|^2$ and boson elastic reflection and transmission probabilities, $|r_{\chi\rightarrow\chi}^B|^2$ and $|t_{\chi\rightarrow\chi}^B|^2$ in the presence of an evanescent bound state, for the graph \texttt{AC(4)}. Dashed line: Corresponding single-particle reflection/transmission probability.}
    \label{fig:appendcyc4 r and t evanescent}
\end{figure}

We see that the bosonic case shows more dramatic change; this is expected by virtue of the symmetric interaction considered. We can see that for the scattering of bosonic particles, the presence of the evanescent state $\ket{\chi_\pm}$ induces perfect \textit{transmission}  at energies near $\pm\sqrt{2}$, in addition to an almost perfect reflection at some energy near $E=\pm 1$. We also note the dramatic asymmetry in the bosonic amplitudes: in the transition to near maximum scattering energy, the transmission probability changes abruptly from near zero to near unity, in the case of the positive energy evanescent bound state; an analogous result holds for the other bound state. This could be exploited as a conditional momentum filter, a gadget that allows only a small range of the momentum wavepacket to be transmitted, depending on the energy of the evanescent bound particle present on the graph.

On the other hand, if the particle bound to the graph is in the \textit{confined} state $\ket{\chi_c}$, then the scattered particle is \textit{perfectly} transmitted for \textit{all} possible energies in the bosonic case, as seen in Fig.\ \ref{fig:appendcyc4 t confined}; this is similar to \textit{infinite hitting times} discussed in \cite{inftyhitting}, where some symmetry of the graph makes a particle unable to visit some vertices. The plots relating to two-particle amplitudes were obtained numerically, but this result can be verified analytically. Naturally, there is no reflection. We also note that, in the case of distinguishable particles, there is elastic transmission and reflection, and also the corresponding processes with the exchange of the particles in the final state; such detail does not manifest in the identical particles case. From now on we focus only on bosonic amplitudes.

\begin{figure}[h]
    \centering
    \includegraphics[width=0.5\textwidth]{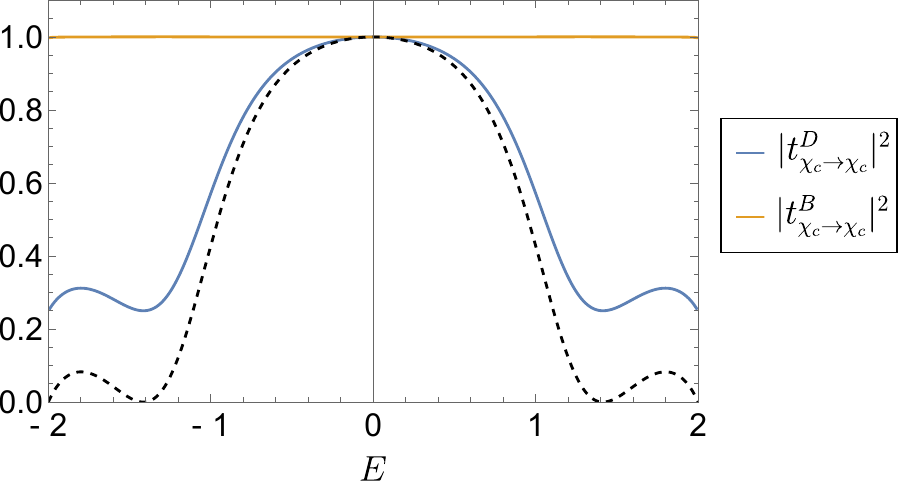}
    \caption[\texttt{AC(4)}: elastic transmission with confined bound state]{Elastic transmission probabilities for distinguishable particles and bosons in the presence of a confined bound state for the graph \texttt{AC(4)}. Dashed line: Single-particle transmission probability.}
    \label{fig:appendcyc4 t confined}
\end{figure}

Let us return to the evanescent case. In Fig.\ \ref{fig:appencyc4 t+r evanescent} a), we see that the sum $|r_{\chi\rightarrow\chi}^B|^2+|t_{\chi\rightarrow\chi}^B|^2$ is \textit{not} $1$ for every incoming momentum, although both processes account for at least $\sim80\%$ of the scattering at each energy. For high (low) incoming energy, we observe that the sum becomes $1$ for the positive (negative) energy bound state; we can interpret this as the impossibility of the particle in the evanescent bound state to lower its energy, since the other particle already has near to maximal (minimal) energy.

\begin{figure}[h]
    \centering
    \includegraphics[width=0.9\textwidth]{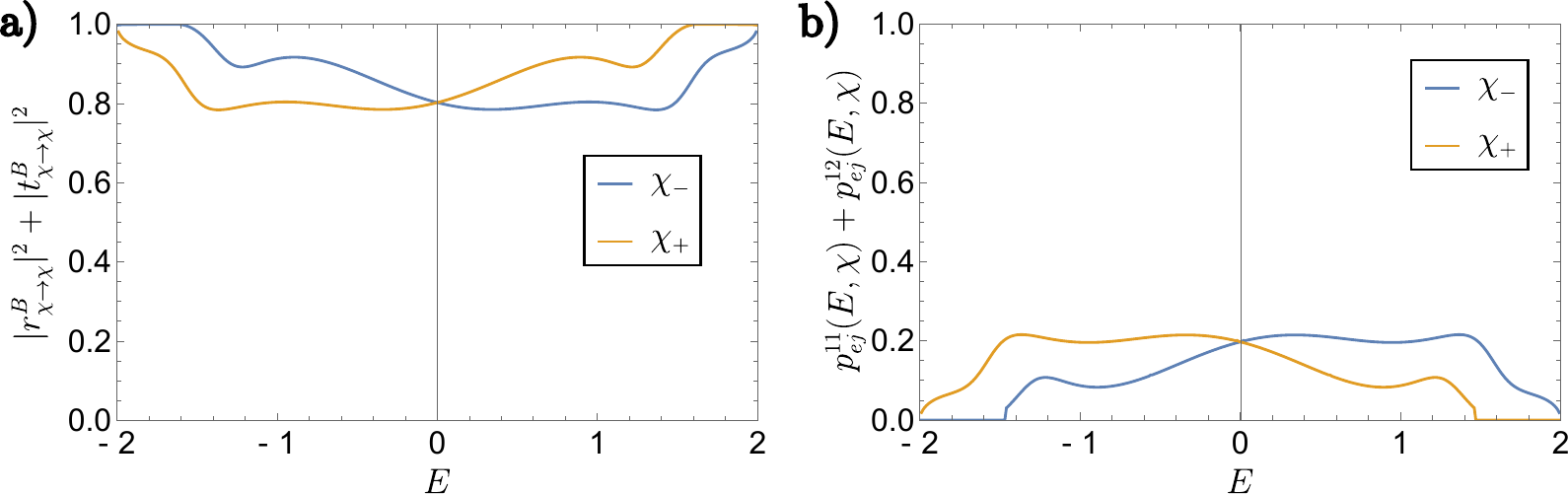}
    \caption[\texttt{AC(4)}: Sum of elastic transmission and reflection probabilities with evanescent bound state]{a) Sum of absolute value squared of  transmission and reflection coefficients in the presence of a evanescent bound state for the graph \texttt{AC(4)}. b) Sum of ejection probabilities of the evanescent bound state, when the scattered particle incomes in rail $1$ and scatters to rail $1$ or $2$ for the graph \texttt{AC(4)}.}
    \label{fig:appencyc4 t+r evanescent}
\end{figure}

So far we have only considered the elastic processes: $\ket{E^{n+}\chi}\rightarrow\ket{E^{m-}\chi}$; for energies where the sum of the probabilities is not $1$, there is another process happening. Let us analyze the \textit{ejection} probabilities of the bound particle: $\ket{E^{n+}\chi}\rightarrow\ket{\mathcal{E}_1^{m_1-}\mathcal{E}_2^{m_2-}}$.

We show in Fig.\ \ref{fig:appencyc4 t+r evanescent} b) the ejection probability in the bosonic case; together with the elastic scattering probabilities they sum up to $1$, and consequently, the full outgoing state corresponds to a combination of elastic scattering and ejection, with no amplitude for inelastic scattering.

Let us now consider the cycle graph with $8$ interacting vertices and rails on symmetric vertices, \texttt{C(8,5)} shown in Fig. \ref{fig:cycle105 single particle graph t r} a). This graph has $4$ different confined bound states $\ket{\chi^{c}_{1\pm}},\ket{\chi^{c}_{2\pm}}$ with energies $\pm E^c_{1}=\pm0.62$ and $\pm E^c_{2}=\pm1.62$, and other $4$ evanescent bound states $\ket{\chi^{ev}_{1\pm}},\ket{\chi^{ev}_{2\pm}}$ with energies $\pm E^{ev}_{1}=\pm2.03$ and $\pm E^{ev}_{2}=\pm2.17$. We also represent in Fig.\ \ref{fig:cycle105 single particle graph t r} the single-particle probabilities $|r(E)|^2$ and $|t(E)|^2$:

\begin{figure}[h]
    \centering
    \includegraphics[width=0.9\textwidth]{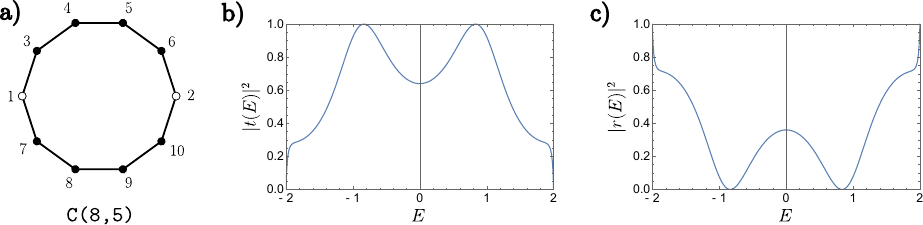}
    \caption[\texttt{C(8,5)}]{From left to right: The graph \texttt{C(8,5)}, the single particle transmission coefficient $|t(E)|^2$ and the single particle reflection coefficient $|r(E)|^2$ as a function of the energy of the incident particle.}
    \label{fig:cycle105 single particle graph t r}
\end{figure}

As before, let us consider the scattering of one incident particle with energy $E$ in the presence of a bound state in the graph, in the bosonic case. In this situation, the \textit{transmission} coefficients are null for \textit{every} incoming momenta and for \textit{any} confined state as well as the ejection probabilities. We now consider the elastic and \textit{inelastic} reflection probabilities. The latter process involves the excitation of the confined bound state to another one: $\ket{E^{n+}\chi}\rightarrow\ket{\mathcal{E}^{m-}\chi'}$.

\begin{figure}[h]
    \centering
    \includegraphics[width=0.9\textwidth]{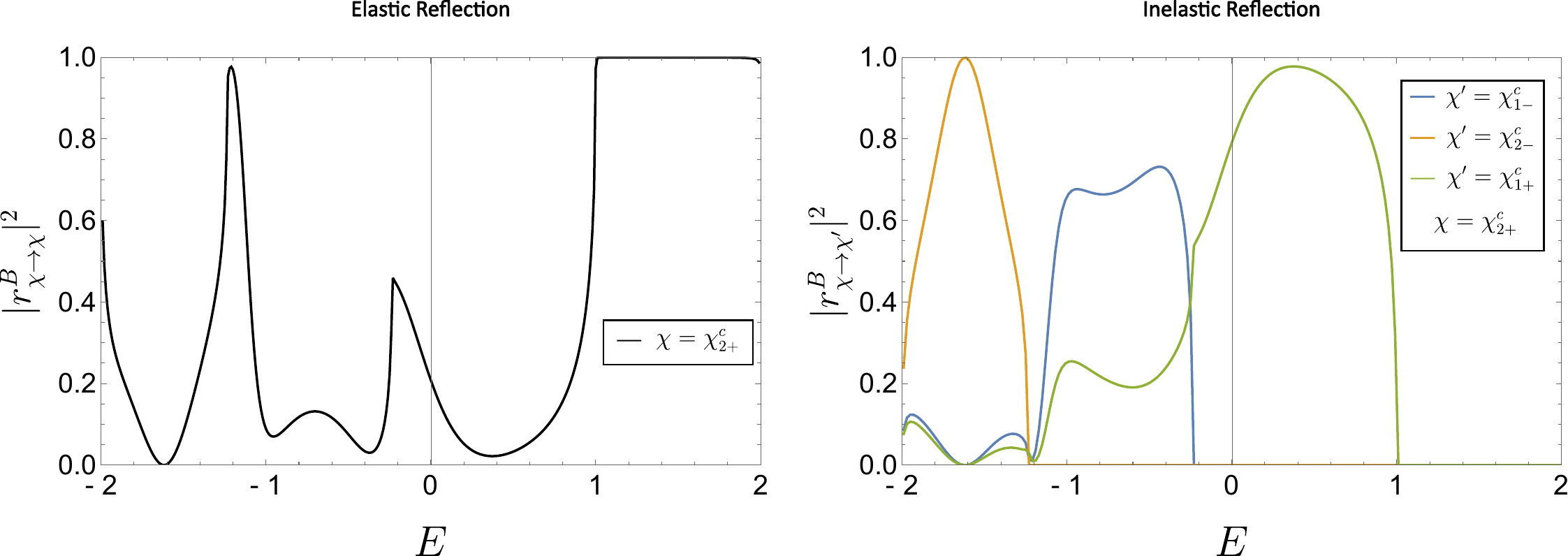}
    \caption[\texttt{C(8,5)}: Elastic and inelastic reflection probabilities with confined bound state]{Left: Probability for elastic reflection with the highest energy confined state $\chi^c_{2+}$ for the graph \texttt{C(8,5)}. Right: Probabilities for inelastic reflection to the other $3$ confined states.}
    \label{fig:cycle105 elastic inelastic r confined}
\end{figure}

In Fig.\ \ref{fig:cycle105 elastic inelastic r confined} we represent the probability of such processes; we chose the starting confined state $\ket{\chi^c_{2+}}$, but we have similar results for the other confined states; we note that these reflection probabilities, elastic and inelastic, sum up to $1$. This situation suggests another type of gadget, one for manipulating confined bound states; depending on the energy of the incoming particle, we can realize different superpositions of the confined bound states. Assuming these results are preserved for bigger cycles, such a gadget is \textit{efficiently scalable} on the number of confined states, since the symmetric cycle with $2N$ vertices has $N$ confined states.

We note that the elastic reflection coefficient is unity for scattering energies greater than $1$. This can be interpreted the same way as before; assuming that there is no inelastic scattering changing from confined to evanescent state, which is confirmed numerically, since the energy gap between the most energetic confined states is $E^c_{2+}-E^c_{1+}=1$, the confined bound state cannot transfer its energy to the incoming particle for $E>1$.

Now let us analyze the situation with evanescent bound states. In Fig.\ \ref{fig:cycle105 r and t evanescent dashed single} we plot the elastic reflection and transmission probabilities for all four evanescent states initialized on the graph. We note that each probability is more intensely modified by the presence of an evanescent particle with greater energy in absolute value. As before, we have $|r_{\chi\rightarrow\chi}^B|^2+|t_{\chi\rightarrow\chi}^B|^2\neq 1$, suggesting that the processes of inelastic reflection and transmission happen with nonzero amplitude.

\begin{figure}[H]
    \centering
    \includegraphics[width=0.9\textwidth]{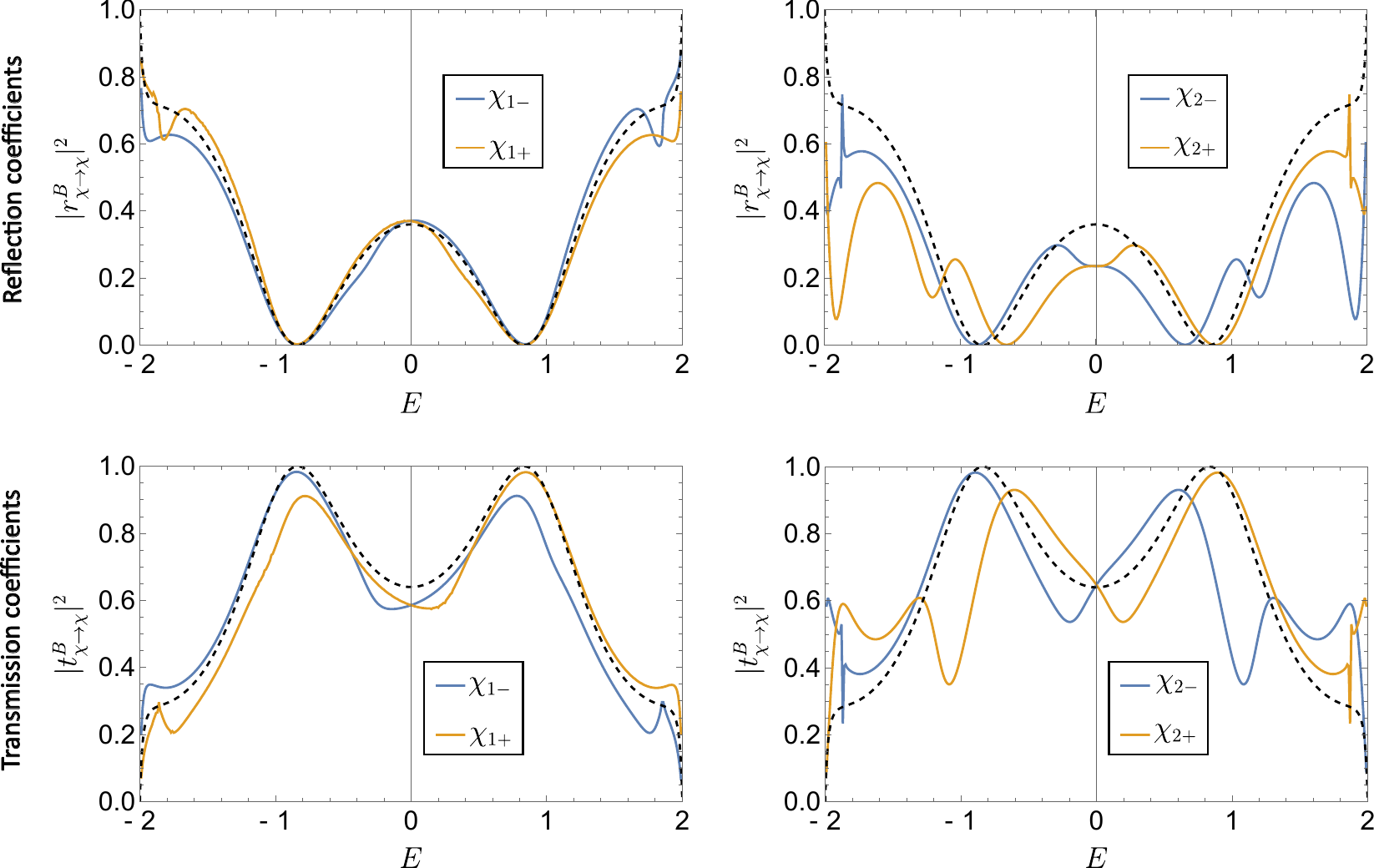}
    \caption[\texttt{C(10,5)}: Elastic reflection and transmission probabilities with evanescent bound state]{From left to right: Bosonic elastic scattering probabilities for states $\chi_{\pm1}$ and $\chi_{\pm2}$ respectively. Dashed line: Corresponding single-particle reflection/transmission probability.}
    \label{fig:cycle105 r and t evanescent dashed single}
\end{figure}

If we consider all such processes, the probabilities still \textit{do not} add up to $1$, but they account for at least $\sim90\%$ of the scattering at each energy, as seen Fig.\ \ref{fig:cycle105 elastic inelastic ejection legend} for the evanescent state with the highest energy $\ket{\chi^{ev}_{2+}}$. In this case, both ejection \textit{and} inelastic reflection/transmission happen simultaneously, in addition to the elastic processes. In Fig.\ \ref{fig:cycle105 elastic inelastic ejection legend} we also plot the probabilities for inelastic scattering and ejection; naturally, we observe that the probabilities for all of those processes do add up to $1$.

\begin{figure}[h]
    \centering
    \includegraphics[width=0.7\textwidth]{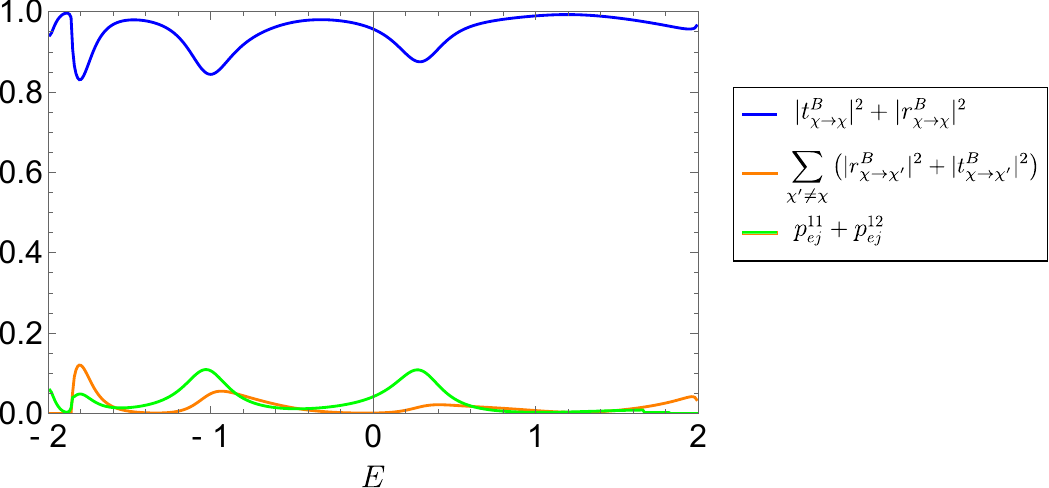}
    \caption[\texttt{C(10,5)}: Elastic and inelastic and ejection probabilities with evanescent bound state]{Probabilities for the elastic, inelastic and ejection processes with initialized bound state $\ket{\chi^{ev}_{2+}}$.}
\label{fig:cycle105 elastic inelastic ejection legend}
\end{figure}

The only processes we have not yet considered are the ones with two free particles in the input state: $\ket{E_1^{n_1+}E_2^{n_2+}}\rightarrow\ket{\mathcal{E}_1^{m_1-}\mathcal{E}_2^{m_2-}}$. These processes have negligible amplitude in the limit of narrow wavepackets in energy space; indeed, since both particles become very delocalized in position space, the interaction between them becomes negligible, and they scatter separately, performing single-particle scattering on the graph. We will get back to this claim in the next section, where we analyze the elastic scattering of two wavepackets with finite width.

\subsection{Traveling states scattering amplitudes}
\label{sec:4C}

We now focus on elastic two-particle scattering of asymptotically free particles; the plots in this section correspond to boson scattering amplitudes. This process depends on the values of the matrix elements $R_{{k_1^{m_1-}}k_2^{m_2-};{p_1^{n_1+}}p_2^{n_2+}}$; we also note that, since these values are proportional to an energy conservation delta factor, the only matrix elements that contribute to actual transition amplitudes satisfy $E_{p_1}+E_{p_2}=E_{k_1}+E_{k_2}$, called on-shell values; therefore we should not strongly believe that specific two-particle scattering properties could be related to off-shell values, since there are infinitely many smooth functions $R$ that have the same value on the support of the delta function. Furthermore, these amplitudes, or their absolute value squared, cannot be interpreted as probabilities directly; however, we expect that the general behavior of the scattering might be related to such amplitudes. Also, for simplicity, we will only plot the real part of $\mathcal{R}$.

Since we only consider graphs with two rails in this section, in addition to being symmetrical graphs with respect to exchanging both rails, let us make the following convention: for  $p\in(-\pi,0]$ we consider a particle incoming with momentum $p$ in rail $1$ and for $p\in(0,\pi]$ we consider a particle incoming with momentum $-p$ in rail $2$.

Let us first consider the \texttt{Line(M)} graph; this situation is analogous to the setting considered in \cite{brod2016a,brod2016b} with photons interacting with sequential quantum dots. The graphs below correspond to $M=39$. We show in Fig.\ \ref{fig:fluor line19 legend} density plots of the real part of $\mathcal{R}_{k_1^{m_1-}k_2^{m_2-};{p_1^{n_1+}}p_2^{n_2+}}$ for fixed incoming momenta, as a function of the outgoing momenta. We chose two pairs of incoming momenta, $p_1=\pm\frac{\pi}{4}$ and $p_2=\frac{\pi}{2}$, corresponding to co-propagating and counter-propagating particles --- other momentum pairs we considered did not present qualitatively differences. We can observe that momentum conservation manifests itself as high values concentrating around the line $k_1+k_2=p_1+p_2$. This behavior has been proved analytically in the limit of $M\rightarrow\infty$ in a previous work \cite{luna}, and it is not totally unexpected, because even though the off-shell values are not related \textit{directly} to the amplitudes of scattering, it is derived from the interaction considered, which is itself translation invariant in the limit of large $M$. 

\begin{figure}[htbp]
    \centering
    \includegraphics[width=0.9\textwidth]{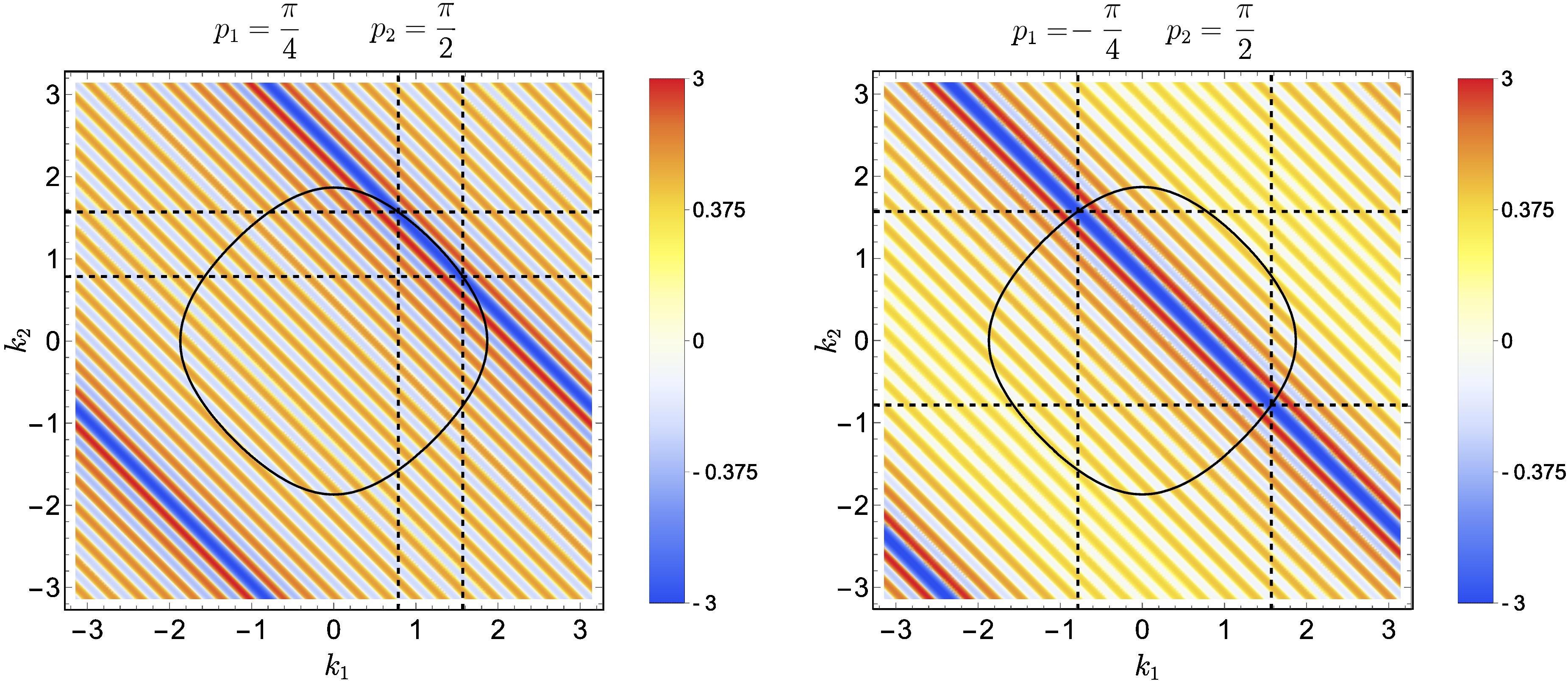}
    \caption[Plots of $\text{Re }R_{k_1k_2;p_1p_2}$ for \texttt{Line(39)}]{Plots of $\text{Re }R_{k_1k_2;p_1p_2}$ for \texttt{Line(39)} for momentum pairs $p_1=\pm\frac{\pi}{4}$, $p_2=\frac{\pi}{2}$ as a function of the outgoing momenta $k_1,k_2$. Left: Co-propagating particles. Right: Counter-propagating particles. Negative and positive values of $k$ represent the particle outgoing onto different rails. Dashed lines: Incoming momenta $p_1$ and $p_2$. Solid line: On-shell points, satisfying $2\cos{k_1}+2\cos{k_2}=2\cos{p_1}+2\cos{p_2}$. To make the structures of the graphs clearer, the color scaling was chosen proportional to $\sqrt[3]{x}$ for these and the next graphs.}
\label{fig:fluor line19 legend}
\end{figure}

In the linear graph, we let each momenta lie in $(-\pi,\pi)$, with the sign representing either a left or right moving particle. The dashed lines represent the values of $p_1$ and $p_2$; we can observe that individual momentum conservation, that is, high amplitudes around $(k_1,k_2)=(p_1,p_2)$ and $(k_1,k_2)=(p_2,p_1)$, manifests itself more rapidly in the counter-propagating case than in the co-propagating case; as we will see later, in the context of calculating cross-sections, the distinction between co- and counter-propagating particles is also observed in other graphs that do not have symmetry by permuting the rails.

\begin{figure}[htpb]
    \centering
    \includegraphics[width=\textwidth]{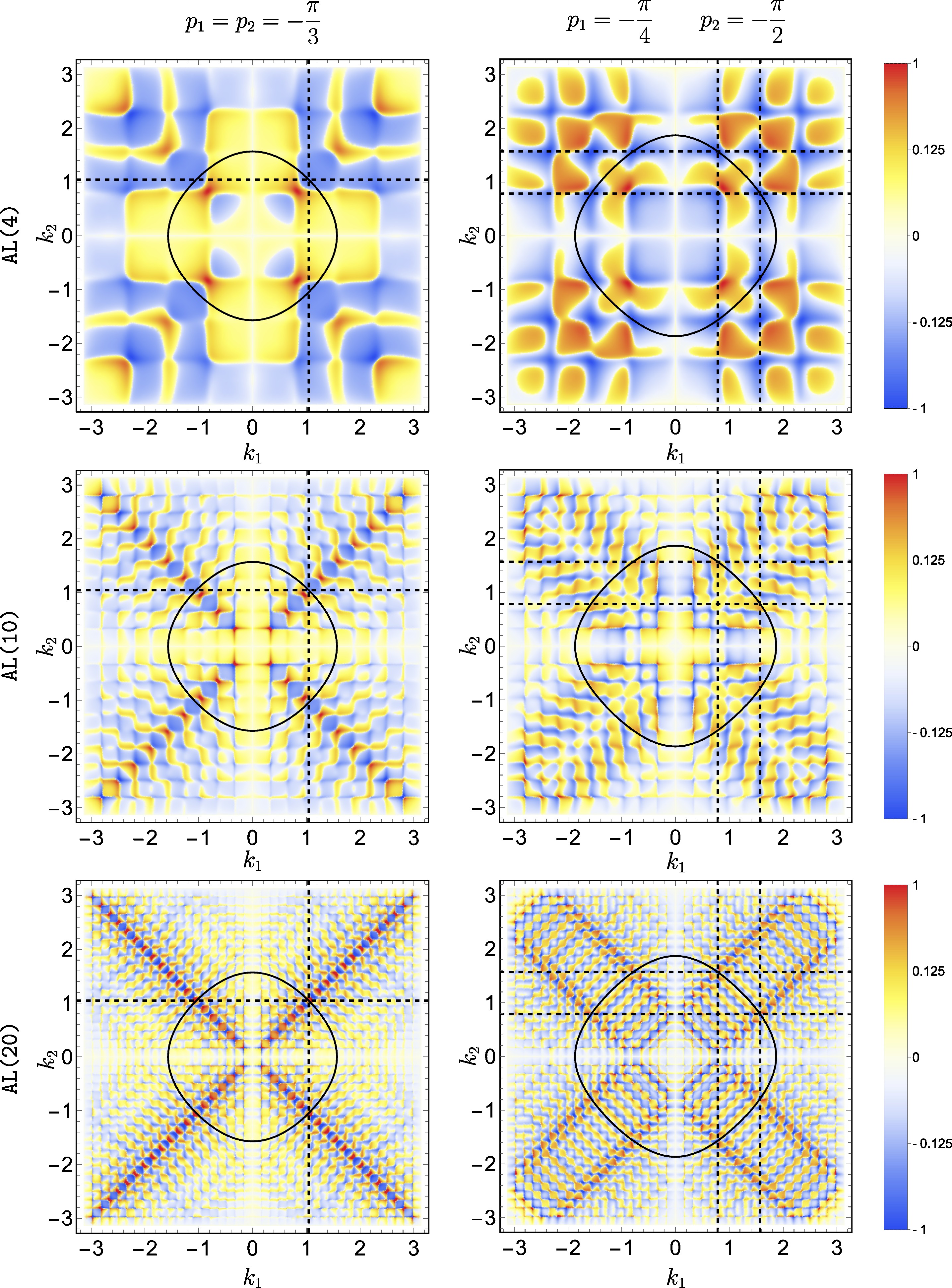}
    \caption[Plots of $\text{Re }R_{k_1k_2;p_1p_2}$ for \texttt{AL(4)}, \texttt{AL(10)} and \texttt{AL(20)}]{Plots of $\text{Re }\mathcal{R}_{k_1k_2;p_1p_2}$ for incoming momentum pairs $p_1=p_2=-\frac{\pi}{3}$ and $p_1=-\frac{\pi}{4}$, $p_2=-\frac{\pi}{2}$ as a function of the outgoing momenta $k_1,k_2$. Negative and positive values of $k$ represent the particle outgoing onto different rails. Dashed lines: Incoming momenta $p_1$ and $p_2$. Solid line: On-shell points, satisfying $2\cos{k_1}+2\cos{k_2}=2\cos{p_1}+2\cos{p_2}$.}
\label{fig:fluor append pi2 pi4 legend}
\end{figure}

Now we consider the appended linear graph \texttt{AL(M)} defined in Appendix \ref{apx:D}; this situation is analogous to another optical system, considered in \cite{appendlineoptics}. We plot in Fig.\ \ref{fig:fluor append pi2 pi4 legend} the amplitudes $\mathcal{R}$ for the pairs $p_1=p_2=-\frac{\pi}{3}$ and $p_1=-\frac{\pi}{4}$, $p_2=-\frac{\pi}{2}$. We can see a $4$-fold symmetry by changing the sign of each individual outgoing momentum; this is consequence of the amplitudes $\mathcal{R}$ depending only on the amplitudes of single-particle free states on the graph, which do not depend on the rail in which the particle is incoming, due to the symmetry of the graph with exchanging vertices $1$ and $2$. This symmetry also holds for the incoming momenta $p_1,p_2$.

We note that the particles have high probability of conserving their momenta, just like in the previous case, although the mechanism that enforces it appears to be different than that of the linear graph. Note that the off-shell values do \textit{not} concentrate on the line $k_1+k_2=p_1+p_2$, but there is still a concentration around the four lines $k_2-k_1=\pm (p_1-p_2)$ and $k_1+k_2=\pm (p_2-p_1)$, specially for bigger $N$. Even though this does not represent momentum conservation \emph{per se}, these equations together with energy conservation imply that either $(k_1,k_2)=(\pm p_1,\pm p_2)$ or $(k_1,k_2)=(\pm p_2,\pm p_1)$, and thus conserving each individual momentum. We may think that momentum conservation arises in the appended linear graph by virtue of the region of interaction becoming more and more translation invariant, in some sense.

In Fig.\ \ref{fig:fluor append pi2 pi4 legend} we can observe a periodic grid pattern with period $\frac{\pi}{M}$ for each graph \texttt{AL(M)} considered. This suggests that particles with such momenta have a more pronounced effect for scattering. To look for these effects, we will consider only the on-shell values of $\mathcal{R}$, on top of the energy conservation curve, which we call $\mathcal{R}|_\mathcal{C}$. Let us now consider the graph \texttt{AL(21)}; in Fig.\ \ref{fig:energyexc append21} we plot the real part of $\mathcal{R}|_{\mathcal{C}}$ as a function of the energy difference $\Delta E=2\cos{k_1}-2\cos{p_1}$ for the input pairs $p_1=-\frac{\pi}{3}$, $p_2=-\frac{\pi}{7}$ and $p_1=-\frac{\pi}{4}$ and $p_2=-\frac{\pi}{8}$. We note that, in general, the appended linear graph favors momentum conservation for specific pairs: we observe narrow peaks for momentum conservation when each incoming momentum is a fraction of $\pi$ whose denominator divides the \textit{number of vertices} of the graph.

\begin{figure}[htbp]
    \centering
    \includegraphics[width=\textwidth]{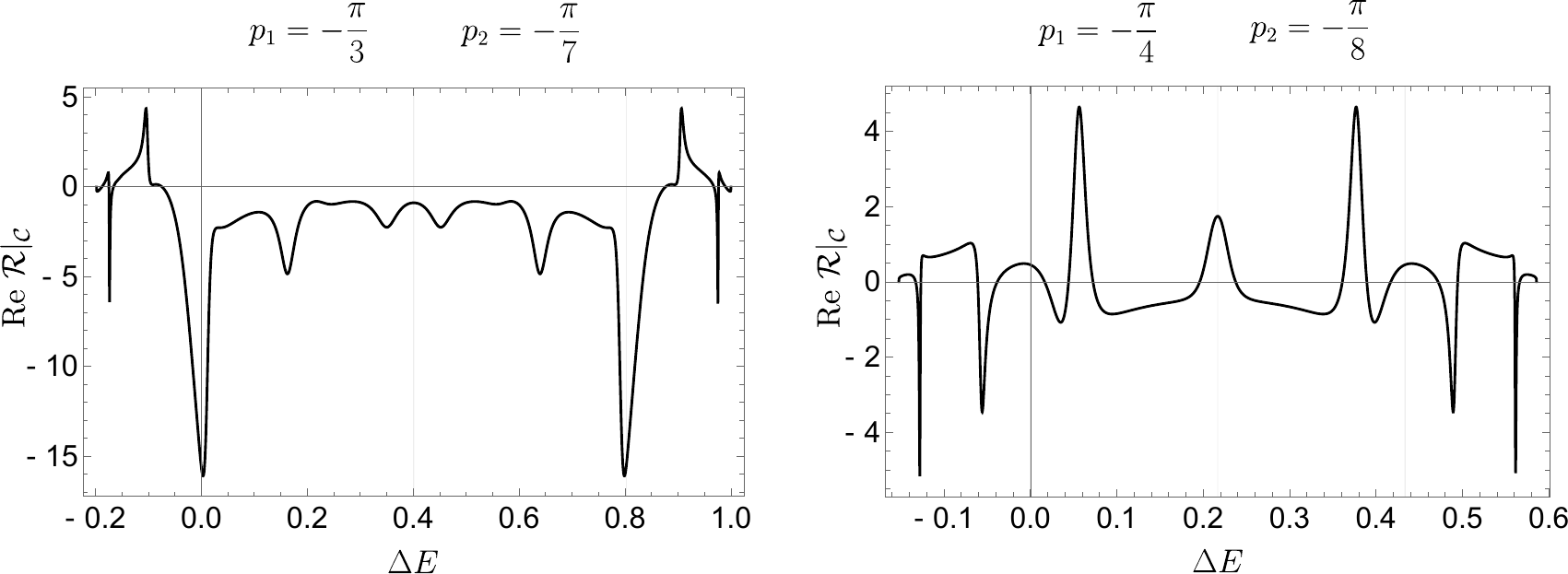}
    \caption[Plots of $\text{Re }\mathcal{R}_{k_1k_2;p_1p_2}$ on top of the energy conservation curve for the graph \texttt{AL(21)}]{Plots of $\text{Re }\mathcal{R}_{k_1k_2;p_1p_2}$ on top of the energy conservation curve for incoming momentum pairs $p_1=-\frac{\pi}{3}$, $p_2=-\frac{\pi}{7}$ and $p_1=-\frac{\pi}{4}$ and $p_2=-\frac{\pi}{8}$ as a function of the energy transferred to the first particle for the graph \texttt{AL(21)}. Grey vertical lines: Left: Point where the energies are equal; Right: Point where the energies are swapped.}
\label{fig:energyexc append21}
\end{figure}

On the left graph we observe narrow peaks at $\Delta E=0$ and $\Delta E=2\cos{\frac{\pi}{7}}-2\cos{\frac{\pi}{3}}$; the first peak corresponds to each particle conserving their individual momenta $(k_1,k_2)=(p_1,p_2)$ and the second peak corresponds to the particles swapping their momenta $(k_1,k_2)=(p_2,p_1)$. For clarity, we represented a gray line where the particles swap their energy, and another line in the midpoint. On the other hand, the right graph does not favor conservation of momentum; for this choice of incoming momenta there is no interesting apparent structure.

\subsection{Two-particle cross section}
\label{sec:4D}

The quantities we have considered so far for studying two-particle scattering do not allow us to unambiguously compare two different given graphs; we sought a new quantity that measures the degree of the interaction between the particles. Before we begin, note that since the S-matrix $S=\mathds{1}+R$ must be unitary, $R$ must satisfy a condition known as the optical theorem:
\begin{equation}
\begin{aligned}
S^\dagger S&=\mathds{1}\\
\rightarrow R+R^\dagger&=-R^\dagger R
\label{opttheo}
\end{aligned}
\end{equation}
On a side note, the optical theorem is usually written in momentum basis, inserting a resolution of the identity between the product in the right-hand side. However, note that, since the free Hamiltonian allows \textit{single-particle} bound states, we need to consider matrix elements of the kind $\ket{\chi{q^n}^+}\bra{\chi{q^n}^+}$ in the resolution of the identity. The complete resolution of the identity in the domain of the S-matrix reads \footnote{We could choose either sign to compute the identity, since $\ket{p^{n+}}$ and $\ket{p^{n-}}$ span the same subspace.} \footnote{If we consider bosons, we should restrict the double integral to $p_1\leq p_2$ whenever $n_1=n_2$.}
\begin{equation}
\begin{aligned}
\mathds{1}=\int_{-\pi}^0\int_{-\pi}^0\sum_{n_1,n_2}\ket{{q_1^{n_1}}^-{q_2^{n_2}}^-}&\bra{{q_1^{n_1}}^-{q_2^{n_2}}^-}dq_1 dq_2+\int_{-\pi}^0\sum_{n,\chi}\ket{\chi{q^n}^-}\bra{\chi{q^n}^-}+\ket{{q^n}^-\chi}\bra{{q^n}^-\chi}dq,
\end{aligned}
\end{equation}
then, the optical theorem read in momentum eigenbasis is written as
\begin{equation}
\begin{aligned}
R_{k_1^{n_1+} k_2^{n_2+};p_1^{m_1+} p_2^{m_2+}}&+R^*_{p_1^{m_1+} p_2^{m_2+};k_1^{n_1+} k_2^{n_2+}}\\
=&-\int_{-\pi}^0\int_{-\pi}^0\sum_{l_1,l_2}R^*_{q_1^{l_1-} q_2^{l_2-};k_1^{n_1+} k_2^{n_2+}}R_{q_1^{l_1-} q_2^{l_2-};p_1^{m_1+} p_2^{m_2+}}dq_1 dq_2\\
&-\int_{-\pi}^0\sum_{l,\chi}R^*_{q^{l-} \chi;k_1^{n_1+} k_2^{n_2+}}R_{q^{l-} \chi;p_1^{m_1+} p_2^{m_2+}}dq\\
&-\int_{-\pi}^0\sum_{l,\chi}R^*_{\chi q^{l-};k_1^{n_1+} k_2^{n_2+}}R_{\chi q^{l-};p_1^{m_1+} p_2^{m_2+}}dq
\end{aligned}
\end{equation}


To benchmark our calculations, we have tested this identity numerically to $10^{-3}$ relative precision for the graphs and momenta considered, by using $\epsilon=10^{-3}$. Choosing smaller values of $\epsilon$ does not  qualitatively change the results, at the cost of being more computationally expensive.

Let us now define the two-particle cross section. Since each particle already scatters on the graph on its own, we want to measure how much the scattering \textit{changes} by the presence of the other particle. Consider the following quantity:
\begin{equation}
\begin{aligned}
\Sigma_{\delta}(E_1^{n_1},E_2^{n_2}):&=||S\ket{\psi_{\delta}(E_1^{n_1},E_2^{n_2})}-\ket{\psi_{\delta}(E_1^{n_1},E_2^{n_2})}||^2\\&=\bra{\psi_{\delta}(E_1^{n_1},E_2^{n_2})}R^\dagger R\ket{\psi_{\delta}(E_1^{n_1},E_2^{n_2})},
\label{2p cross section}
\end{aligned}
\end{equation}
where $\ket{\psi_{\delta}(E_1^{n_1},E_2^{n_2})}$ is a product of two square wavepackets centered about $E_1$ and $E_2$, each with width $2\delta$, on rails $n_1$ and $n_2$ respectively. Note that $0\leq\Sigma_\delta\leq 4$, since it is the norm of a difference of two normalized vectors. Applying the optical theorem, we have\footnote{Again, if we consider bosons, we should restrict the double integral to $\mathcal{E}_1\leq \mathcal{E}_2$ and $\mathcal{E'}_1\leq \mathcal{E'}_2$ if $n_1=n_2$.}
\begin{equation}
\begin{aligned}
\Sigma_{\delta}(E_1^{n_1},E_2^{n_2})&=-2\,\text{Re}\bra{\psi_{\delta}(E_1^{n_1},E_2^{n_2})}R\ket{\psi_{\delta}(E_1^{n_1},E_2^{n_2})}\\
&=-\frac{2}{4\delta^2}\text{Re}\int_{E_1-\delta}^{E_1+\delta}\int_{E_2-\delta}^{E_2+\delta}\int_{E_1-\delta}^{E_1+\delta}\int_{E_2-\delta}^{E_2+\delta} R_{\mathcal{E}_1^{n_1+} \mathcal{E}_2^{n_2+};\mathcal{E'}_1^{n_1+} \mathcal{E'}_2^{n_2+}} d\mathcal{E}_1d\mathcal{E}_2d\mathcal{E}'_1d\mathcal{E}'_2\\
&=:-2\text{ Re }I(E_1^{n_1},E_2^{n_2})
\end{aligned}
\end{equation}

Note that, since $R$ has one Dirac delta factor, the full integral can be written as a triple integral of a continuous function on a volume proportional to $\delta^3$, which renders $\Sigma_{\delta}(E_1^{n_1},E_2^{n_2})=O(\delta)$ for $\delta\rightarrow0$. This implies that the interaction becomes negligible for very narrow wavepackets. In the setting of non-linear optics, in \cite{brod2016a}, Brod and Combes discuss how for any finite-sized region there is a value of wavepacket (spectral) width that is small enough for the nonlinear effect to disappear. This is corroborated by results from our previous paper \cite{luna} and is a manifestation of the cluster decomposition principle, discussed in \cite{clusterdecprinc}.

Also, since $\ket{\psi_{\delta}(E_1^{n_1},E_2^{n_2})}$ represents incoming particles, the amplitude inside the integral relates incoming ($\mathcal{E'}_1^{n_1+} ,\mathcal{E'}_2^{n_2+}$) states to \textit{incoming} ($\mathcal{E}_1^{n_1+} ,\mathcal{E}_2^{n_2+}$) states; this situation is different than the one considered in the last section, where we were interested in transition amplitudes from incoming to outgoing free states. Note also that, since the matrix elements of $R$ in the bosonic case are double that of distinguishable particles for the Hubbard interaction, typically this will also be the case for the cross section, except when $|E_2-E_1|\leq\delta$ and $n_1=n_2$.

This integral is similar to the fidelity computed in \cite{luna} in the linear graph setting; both the fidelity $F(\phi)$ of the CPHASE gate and the two-particle cross section defined here depend on the same integral $I$, with the only difference that we considered wavepackets in \textit{momentum} space before. In the former situation, the fidelity with a CZ \footnote{The controlled Z gate is a CPHASE with phase $-1$.} is maximal if the integral evaluates to $I=-2$, which corresponds to maximal two-particle cross section $\Sigma_\delta=4$.

\begin{figure}[htbp]
    \centering
    \includegraphics[width=0.5\textwidth]{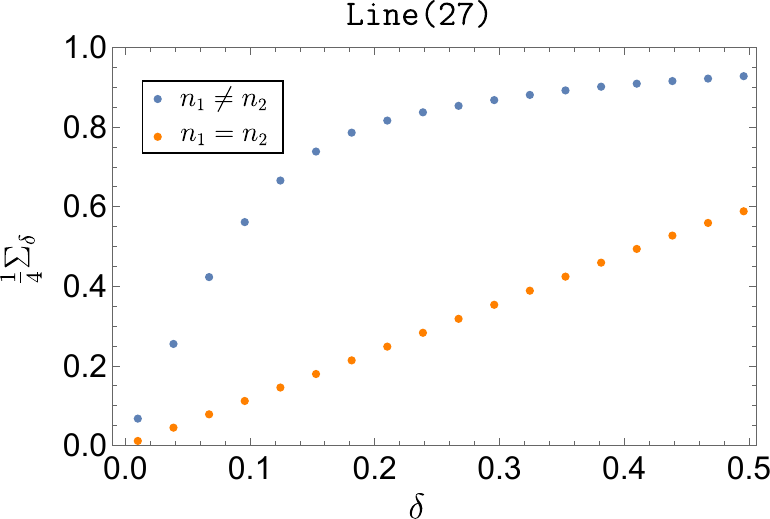}
    \caption[Two-particle cross section for \texttt{Line(27)}]{Bosonic two-particle cross section $\Sigma_\delta(E_1^{n_1},E_2^{n_2})$ for the linear graph \texttt{Line(27)} with $E_1=0$ and $E_2=\sqrt{2}$ for counter- and co-propagating particles.}
\label{fig:cross section line 27 co counter legend}
\end{figure}

Let us analyze the two-particle cross section for some graphs. The usefulness of such a quantity is that it summarizes all effects of the interaction between the scattering particles in a single figure of merit. Going back to the linear graph, we see that the cross section is already close to maximal in the counter-propagating case with few interaction sites, as we see in Fig.\ \ref{fig:cross section line 27 co counter legend}. This is equivalent to saying that the linear graph is approximating the CZ gate as in the setting of \cite{luna}. Here and in the next example, we observe that counter-propagating particles achieve higher cross sections, and in some sense, have higher chance of interacting.

\begin{figure}[htbp]
    \centering
    \includegraphics[width=\textwidth]{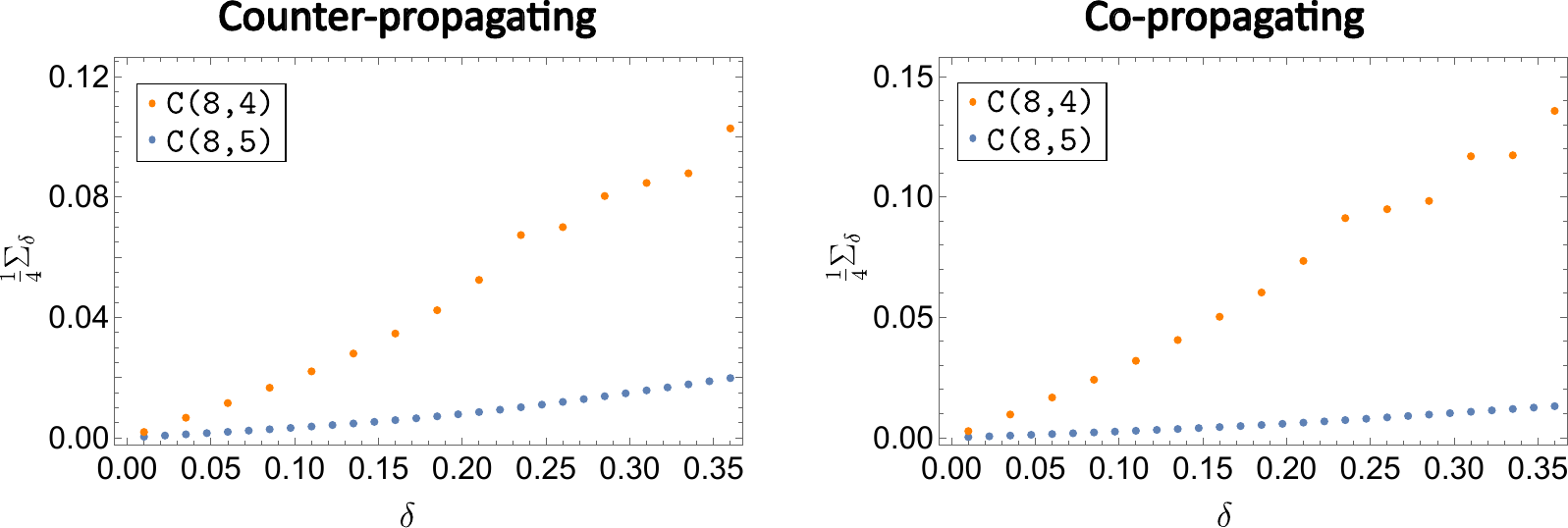}
    \caption[Two-particle cross sections for \texttt{C(8,5)} and \texttt{C(8,4)}]{Two-particle cross section $\Sigma_\delta(E_1^{n_1},E_2^{n_2})$ for the cycles \texttt{C(8,5)} and \texttt{C(8,4)}, with $E_1=0$, $E_2=\sqrt{2}$. Left: Counter-propagating ($n_1\neq n_2$). Right: Co-propagating ($n_1= n_2$) .}
\label{fig:cross section C84 C83 final}
\end{figure}

Now let us turn our attention to other graphs; in Fig.\ \ref{fig:cross section C84 C83 final} we plot the bosonic two-particle cross section for the symmetric cycle \texttt{C(8,5)} and the asymmetric cycle \texttt{C(8,4)}. We had the intuition that big symmetric cycles would behave similarly to the long line \texttt{Line(M)} asymptotically, because the particles would spend most of the time traveling on either half of the cycle. However, we observe that the asymmetric graph \texttt{C(8,4)} favors larger cross sections; in this case, we have an increase of \textit{one order of magnitude}. This would imply that, to build a gadget that exploits free particle scattering, we might want to look for asymmetric graphs. We also note that asymmetry may imply a smaller number of bound states on the graph, as seen in Table \ref{tab:bound}, which may be related to larger cross sections; we leave further investigation of this connection as an open question. Every graph we considered, except for the line, has cross section of order $10^{-1}$ to $10^{-2}$; it would be interesting to find other examples of graphs attaining high cross sections.

We note that the two-particle cross section takes into account both the elastic scattering of the free particles and the time reversal of the ejection, i.e., the capture of one of the free particles by the graph. In the last section, for the graphs considered, we have seen that ejection seems to account for only a small part of the full scattering; this may be a consequence of the short range of the interaction we considered.

\section{Conclusion} \label{sec:conclusion}

Our systematic approach to studying two-particle scattering on graphs lays the groundwork for further development in several directions. In a more practical point of view, we can apply these new tools to different physical systems; the quantum-walk Hamiltonian of a particle on a graph is analogous to the tight-binding model in condensed matter physics, and our results could be applied to studying properties of electron transport in more general lattices, such as the hexagonal lattice of graphene; in optics, for example, we could apply quantum walk with the Hubbard interaction on specified vertices as a toy model to study photons interacting with a quantum dot network --- we found strong similarities of scattering properties of simple families of graphs and the optical systems considered in \cite{brod2016a,brod2016b,appendlineoptics}.

On the other hand, we believe these results can shed new light on studying new graph theoretical properties in general; for example, some frequently studied properties in general single-particle walks are \textit{hitting times} --- the average time it takes for a particle to traverse between two given vertices --- and \textit{mixing times} --- the average time it takes for the distribution to become stationary. We can ask similar questions in \textit{two-particle} quantum walks: 
\begin{itemize}
    \item \textbf{Two-particle hitting times:} What is the average time it takes for a particle to traverse a given graph when a second particle is in some bound state? What is the average time it takes for two particles initialized in some inner vertex to escape the graph? Some of our results in section \ref{sec:4B} are similar to \textit{infinite hitting times} discussed in \cite{inftyhitting}.
    \item \textbf{Two-particle mixing times:} How long does it take for the bound particle to change to a different bound state in an inelastic scattering process? In opposition to the dynamics on finite graphs, scattering on graphs with rails is \textit{not} periodic; it is possible to define asymptotic mixing times.
\end{itemize}

In the perspective of quantum computation, we can explore new possible gadgets for quantum algorithms based on graphs, or more generally, quantum tasks. The graphs described in section \ref{sec:4B} already exhibit many new properties that could be interesting: We can describe \textit{how} the presence of a bound state on part of a graph can affect the trajectory of another particle traversing it. The graph \texttt{AC(4)} allows for a \textit{conditional momentum filter}, a gadget
that allows only a small range of the momentum wavepacket to be transmitted, depending on
the energy of the evanescent bound particle present on the graph. In addition, by considering confined bound states, the same graph can serve as a kind of \textit{transistor}, making a particle perfectly transmit or reflect through the graph depending on the presence of the bound particle. Also, even though we do not expect it is possible to construct a finite-sized graph which implements a perfect CPhase gate between two traveling particles \cite{clusterdecprinc}, it is still open whether it is possible to intermediate an entangling gate using bound states; we leave this as an interesting open question.

In the perspective of complexity theory, we opened the way to answer some open questions: since multi-particle interacting quantum walks are universal for quantum computing \cite{singleparticleqw,multiparticleqw}, can a restricted form of quantum walks have intermediate computational power?
\begin{enumerate}
    \item \textbf{No interaction}: In analogy to the boson-sampling model \cite{bossamp}, are quantum walks with non-interacting bosons hard to simulate classically? Expanding the results in \cite{FLO}, can we show that quantum walks with non-interacting \textit{fermions} are classically simulable?
    \item \textbf{Restricting momentum}: Are quantum walks where all particles have the same initial momentum universal for quantum computing? In the original construction it is necessary to have particles with two different momenta interacting.
    \item \textbf{Restricting gadgets}: Is there a restricted set of gadgets which makes quantum walks non-universal? In the original construction, the gadget which implements the Hadamard gate has maximal degree $4$, therefore restricting the graph's maximum degree to $3$ could render quantum walks non-universal.
    \item \textbf{Bound states as resources}: Given a non-universal set of gadgets, can initializing the graph with some number of bound particles make it universal?
\end{enumerate}

We applied our method only to the Hubbard model, but the formalism is easily applicable to other multi-particle potentials.  

\begin{acknowledgments}
This work was supported by Brazilian funding agencies CNPq and FAPERJ.
\end{acknowledgments}

\bibliography{References}

\appendix

\section{Single-particle S-matrix}
\label{apx:A}
In this section, we compute an analytic expression for the single-particle S-matrix.

Starting from Eq.\ (\ref{eq:singleparticleS}), we can write the factor relating to the interaction as:
\begin{equation}
\begin{aligned}
\bra{p^m}V\ket{{p^n}^+}=\sum_{u\in V(G)}\dfrac{z-z^{-1}}{\sqrt{2\pi}}\bra{m}H_G\ket{u}\braket{u|{p^n}^+},
\label{pVp}
\end{aligned}
\end{equation}
since $\ket{{p^m}}$ only overlaps with the graph $G$ at the vertex $\ket{m}=\ket{1,m}$.

The terms $\braket{u|{p^n}^+}$ can be found by taking the inner product of the Lippmann-Schwinger equation with $\bra{u}$ for $u\in V(G)$:
\begin{equation}
\begin{aligned}
\braket{u|{p^n}^+}&=\braket{u|{p^n}}+\sum_{v\in V(G) }\bra{v}H_G\ket{{p^n}^+}\left(\sum_{t=1}^N\int_{-\pi}^0\dfrac{\braket{u|k^t}\braket{k^t|v}}{E-2\cos{k}+i\epsilon}dk+\sum_{w\in G^0}\dfrac{\braket{u|w}\braket{w|v}}{E+i\epsilon}\right)\\
&=:\dfrac{e^{-ip}-e^{ip}}{\sqrt{2\pi}}\delta_{un}+\sum_{v\in V(G)} J_{uv}\bra{v}H_G\ket{{p^n}^+},
\label{1plipp}
\end{aligned}
\end{equation}
where we used the resolution of the identity $\mathds{1}=\sum_{t=1}^N\int_{-\pi}^0 \ket{k^t}\bra{k^t}dk+\sum_{w\in G^0}\ket{w}\bra{w}$. Changing the integration variable to $\omega=e^{ik}$ and writing $E=z+z^{-1}$, we have
\begin{equation}
\begin{aligned}
J_{uv}&=\sum_{t=1}^N\dfrac{1}{2}\oint\dfrac{-\delta_{ut}\delta_{tv}(\omega^2+\omega^{-2}-2)}{E-\omega-\omega^{-1}+i\epsilon}\dfrac{d\omega}{2\pi i\omega}+\sum_{w\in G^0}\dfrac{\braket{u|w}\braket{w|v}}{E+i\epsilon}\\
&=\delta_{uv}\mathds{1}_{u\in\partial G}\oint\dfrac{\omega^2-1}{(\omega-\omega^-)(\omega-\omega^+)}\dfrac{d\omega}{2\pi i}+\dfrac{\delta_{uv}\mathds{1}_{u\notin\partial G}}{E+i\epsilon}\\
&=\omega^-\delta_{uv}\mathds{1}_{u\in\partial G}+\dfrac{\delta_{uv}\mathds{1}_{u\notin\partial G}}{z+z^{-1}+i\epsilon},
\end{aligned}
\end{equation}
where $\mathds{1}_\mathcal{C}=1$ if condition $\mathcal{C}$ is true and $0$ otherwise, and 
$$\omega^\pm:=\frac{1}{2}\left(E+i\epsilon\pm\sqrt{E+2+i\epsilon}\sqrt{E-2+i\epsilon}\right)$$
are the roots of the denominator of the integrand\footnote{We choose the principal branch of the square root.}. We also used the fact that $\oint f(\omega)\frac{d\omega}{\omega}=\oint f(\omega^{-1})\frac{d\omega}{\omega}$ for integrals on the unit circle. We have $|\omega^-|<1<|\omega^+|$, so the  only pole inside the contour is $\omega^-$. Moreover, the roots of $E-\omega-\omega^{-1}$ are exactly $z$ and $z^{-1}$, so in the limit $\epsilon\searrow 0$, we have $\omega^\pm\rightarrow z^{\mp 1}$ and the representation of the product $H_G J$ in position basis is given by the matrix\footnote{Since $p\in(-\pi,0)$, the imaginary part of $z$ is negative, so it corresponds to $\omega^-$.}:
\begin{equation}
\begin{aligned}
H_G J=\begin{bmatrix}
    \omega^-A & \omega^-B^\dagger \\
    \dfrac{B}{z+z^{-1}+i\epsilon} & \dfrac{D}{z+z^{-1}+i\epsilon}
\end{bmatrix}.
\end{aligned}
\end{equation}

Defining $T=\mathds{1}-H_G J$, we can rewrite \eqref{1plipp} as
\begin{equation}
\begin{aligned}
\braket{u|p;n^+}=\sum_{v\in V(G)}T^{-1}_{uv}\dfrac{e^{-ip}-e^{ip}}{\sqrt{2\pi}}\delta_{vn}=\dfrac{z^{-1}-z}{\sqrt{2\pi}}T^{-1}_{un},
\label{scp}
\end{aligned}
\end{equation}
where the inverse of $T$ can be computed as a block matrix inverse
\begin{equation}
\begin{aligned}
T^{-1}&=\begin{bmatrix}
    Q_{\epsilon}(z)^{-1} & *  \\
    (z+z^{-1}+i\epsilon-D)^{-1}BQ_{\epsilon}(z)^{-1}   & *
\end{bmatrix},
\label{T}
\end{aligned}
\end{equation}
and $*$ are unimportant matrices, since we only need the first $N$ columns of $T^{-1}$, and $Q_{\epsilon}(z)=\mathds{1}-\omega^-A-\omega^-B^\dagger(z+z^{-1}+i\epsilon-D)^{-1}B$.

Going back to Eq.\ \eqref{pVp}:
\begin{equation}
\begin{aligned}
\bra{p^m}V\ket{{p^n}^+}&=\dfrac{z-z^{-1}}{\sqrt{2\pi}}\sum_{u\in V(G)}\bra{m}H_G\ket{u}(T^{-1})_{un}\dfrac{z^{-1}-z}{\sqrt{2\pi}}\\
&=-\dfrac{1}{2\pi}(z-z^{-1})^2\left(\begin{bmatrix}
    A & B^\dagger  \\
    B & D
\end{bmatrix} T^{-1}\right)_{mn}.
\end{aligned}
\end{equation}
Finally the S-matrix at fixed momentum is
\begin{equation}
\begin{aligned}
\mathbb{S}(z)&=\lim_{\epsilon\searrow 0}-\mathds{1}+\dfrac{2\pi}{z-z^{-1}}\dfrac{(z-z^{-1})(z^{-1}-z)}{2\pi}(A+B^\dagger(z+z^{-1}+i\epsilon-D)^{-1}B)Q_{\epsilon}(z)^{-1}\\
&=\lim_{\epsilon\searrow 0}-Q_{\epsilon}(z)Q_{\epsilon}(z)^{-1}+(z^{-1}-z)(A+B^\dagger(z+z^{-1}+i\epsilon-D)^{-1}B)Q_{\epsilon}(z)^{-1}\\
&=\lim_{\epsilon\searrow 0}-\left((\mathds{1}+(z-\omega^--z^{-1})(A+B^\dagger(z+z^{-1}+i\epsilon-D)^{-1}B)\right)Q_{\epsilon}(z)^{-1}\\
&=-Q(z^{-1})Q(z)^{-1},
\end{aligned}
\end{equation}
where $Q(z)=\mathds{1}-zA-zB^\dagger(z+z^{-1}-D)^{-1}B$. This process may render the matrix $Q$ not invertible for some values of $z$; for a discussion see \cite{varbanov}. This is exactly the result described in \cite{levinsonthm2}, obtained by a different method. 

Looking back to equation \eqref{scp}, we note that the amplitudes of the scattering state going in through rail $n$ on the graph is the $n$-th column of the matrix
\begin{equation}
\begin{aligned}
\frac{1}{\sqrt{2\pi}}\begin{bmatrix}
    (z^{-1}-z)Q(z)^{-1} \\
    (z^{-1}-z)(z+z^{-1}-D)^{-1}BQ(z)^{-1}
\end{bmatrix}&=\frac{1}{\sqrt{2\pi}}\begin{bmatrix}
    z^{-1}\mathds{1}+z\mathbb{S}(z) \\
    (z+z^{-1}-D)^{-1}B(z^{-1}\mathds{1}+z\mathbb{S}(z))
\end{bmatrix},
\end{aligned}
\end{equation}
i.e., the boundary amplitudes are $\braket{m|p^{n+}}=(z^{-1}\delta_{mn}+z\mathbb{S}(z)_{mn})/\sqrt{2\pi}$ for $m\in\partial G$, and the inner amplitudes $\braket{u|{p^n}^+}$ for $u\in G^0$ is the $(u,n)$-matrix element of
\begin{equation}
\begin{aligned}
\frac{1}{\sqrt{2\pi}}\Psi(z):=\frac{1}{\sqrt{2\pi}} (z+z^{-1}-D)^{-1}B(z^{-1}\mathds{1}+z\mathbb{S}(z)).
\label{Psi}
\end{aligned}
\end{equation}

We have considered until now a scattering state that, for $t\rightarrow-\infty$, is incoming on rail $n$. Computing the amplitude of the scattering states which are outgoing on rail $n$ for $t\rightarrow\infty$ is similar, except for the values of $\omega^\pm$. It is easy to check that, for $-k,p\in(-\pi,0)$, the second equation also holds:
\begin{equation}
\begin{aligned}
\braket{u|{p^n}^+}&=\frac{1}{\sqrt{2\pi}}\Psi_{un}(e^{ip})\\
\braket{u|{k^n}^-}&=\frac{1}{\sqrt{2\pi}}\Psi_{un}(e^{ik}).
\end{aligned}
\end{equation}

The matrices $\Psi$ and $\mathbb{S}$ give a complete description of the single particle scattering states. As an example, let us prove that 
\begin{equation}
\begin{aligned}
\braket{x;m|p^{n+}}&=\dfrac{1}{\sqrt{2\pi}}(\delta_{mn}e^{-ipx}+\mathbb{S}_{mn}e^{ipx}) &\text{ for } x\ge 1 \text{ and } m=1,\cdots,N.
\label{A12}
\end{aligned}
\end{equation}

Applying the full Hamiltonian, which is the adjacency matrix of the full graph, to $\ket{p^{n+}}$ on the boundary vertex $\bra{1;n}$, we have
\begin{equation}
\begin{aligned}
\bra{m}H\ket{p^{n+}}&=\braket{2;m|p^{n+}}+\bra{m}H_G\frac{1}{\sqrt{2\pi}}\begin{bmatrix}
    z^{-1}\mathds{1}+z\mathbb{S}(z) \\
    (z+z^{-1}-D)^{-1}B(z^{-1}\mathds{1}+z\mathbb{S}(z))
\end{bmatrix}\ket{n}.
\end{aligned}
\end{equation}

Since $\ket{p^{n+}}$ is an eigenstate with energy $2\cos{p}=z+z^{-1}$, and using $H_G=\begin{bmatrix}
    A & B^\dagger  \\
    B & D
\end{bmatrix}$, we have
\begin{equation}
\begin{aligned}
\braket{2;m|p^{n+}}&=(z+z^{-1})\braket{m|p^{n+}}-(A(z^{-1}\mathds{1}-z\mathbb{S}(z))+B^\dagger\Psi)\\
&=(z+z^{-1})(z^{-1}\mathds{1}-z\mathbb{S})_{mn}-(A(z^{-1}\mathds{1}-z\mathbb{S}(z))+B^\dagger(z+z^{-1}-D)^{-1}B(z^{-1}\mathds{1}+z\mathbb{S}(z)))_{mn}\\
&=[((z+z^{-1})\mathds{1}-A-B^\dagger(z+z^{-1}-D)^{-1}B)(z^{-1}\mathds{1}+z\mathbb{S}(z))]_{mn}\\
&=(z^{-2}Q(z)+Q(z)\mathbb{S}(z)+\mathds{1}+z^2\mathbb{S}(z))_{mn}.
\end{aligned}
\end{equation}

Since $\mathbb{S}(z)=-Q(z^{-1})Q(z)^{-1}$ and $[Q(z^{-1}),Q(z)]=0$, we can simplify:
\begin{equation}
\begin{aligned}
\braket{2;m|p^{n+}}=(z^{-2}\mathds{1}+z^2\mathbb{S})_{mn},
\end{aligned}
\end{equation}
which corresponds to Eq.\eqref{A12} for $x=2$. All the other amplitudes are now uniquely determined by the Schrodinger equation on each rail, which has the general solution $Ae^{-ipx}+Be^{ipx}$; this concludes the proof of Eq.\eqref{A12} for all $x\geq 1$.

In the next Section we describe the bound states of this Hamiltonian. 

\section{Single-particle bound states on general graphs}
\label{apx:B}
We now look for normalizable solutions to the Schr{\"o}dinger equation. There are two kinds of such states; we will start describing the confined bound states, which are the states with zero amplitude on the rails. In other words, they have finite support and therefore are automatically normalizable.

Since they only have nonzero amplitude on the inner vertices, they satisfy the eigenvalue equation on $G$:
\begin{equation}
\begin{aligned}
\begin{bmatrix}
    A & B^\dagger  \\
    B & D
\end{bmatrix}\begin{bmatrix}
    0 \\
    \vec{\beta_c} 
\end{bmatrix}=\lambda_c\begin{bmatrix}
    0 \\
    \vec{\beta_c} 
\end{bmatrix}=:\lambda_c\ket{\psi_c},
\end{aligned}
\end{equation}
where $\vec{\beta_c}$ is a column vector containing the $M$ amplitudes of the confined bound state $\ket{\psi_c}$ on $G^0$. The confined bound states are then characterized by $D\vec{\beta_c}=\lambda_c\vec{\beta_c}$ and $B^\dagger\vec{\beta_c}=0$. Let us denote this subspace by $\mathcal{C}$.

Let us now look for other bound states in the orthogonal complement $\mathcal{C}^\perp$.  We have seen that, on the infinite line, the possible solutions are of the form $\psi(x)=z^x$, $x\in\mathbb{Z}$, and no solution is normalizable.  On the other hand, since we are now considering semi-infinite lines, exponential decreasing solutions are admissible. These states have been called unconfined, or evanescent, bound states in the literature. So, we are looking for states of the form
\begin{equation}
\begin{aligned}
\braket{x;j|\phi}=\alpha_j z^{x-1},
\label{otheramp}
\end{aligned}
\end{equation}
on each rail $j\in{1,\cdots,N}$, with $|z|<1$; the authors of \cite{levinsonthm2} show that the only possibility is $z\in(-1,1)$. Writing $\vec{\alpha}$ for the first $N$ amplitudes and $\vec{\beta}$ for the other $M$ amplitudes, the Schr{\"o}dinger equation for the amplitudes on the graph is written in matrix form as
\begin{equation}
\begin{aligned}
\begin{bmatrix}
    A & B^\dagger \\
    B & D
\end{bmatrix}\begin{bmatrix}
    \vec{\alpha}  \\
    \vec{\beta}
\end{bmatrix}+z\begin{bmatrix}
    \vec{\alpha}  \\
    0
\end{bmatrix}=(z+z^{-1})\begin{bmatrix}
    \vec{\alpha}  \\
    \vec{\beta}
\end{bmatrix}.
\label{QEP0}
\end{aligned}
\end{equation}

It is convenient to rewrite this equation as a quadratic eigenvalue problem \footnote{See \cite{QEP} for a review on QEP's.}:
\begin{equation}
\begin{aligned}
\gamma(z)\begin{bmatrix}
    \vec{\alpha}  \\
    \vec{\beta}
\end{bmatrix}=0,
\label{QEP}
\end{aligned}
\end{equation}
where $\gamma(z)$ is the operator
\begin{equation}
\begin{aligned}
\gamma(z)=\begin{bmatrix}
    zA-\mathds{1} & zB^\dagger  \\
    zB & zD-(z^2+1)\mathds{1}
\end{bmatrix}=-z^2 P_M + z H_G -\mathds{1},
\label{Gamma}
\end{aligned}
\end{equation}
and $P_M$ is the matrix for the projector on the $M$ inner vertices. The operator $\gamma(z)$ is called a matrix polynomial; in the ordinary eigenvalue problem, we are given a matrix $A$ and we look for a vector $v$ and a value $z\in\mathbb{C}$ such that $(A-z\mathds{1})v=0$. In the generalized eigenvalue problem we have a matrix polynomial $B(z)$ and we define the eigenvectors as solutions to $B(z)v=0$ for some specified $z\in\mathbb{C}$. Since it is not a very common problem some comments are in order; let us consider a matrix polynomial of degree 2 and size $n\times n$:

First, as in the ordinary eigenvalue problem, the possible eigenvalues of $B(z)=B_0+zB_1 +z^2 B_2$ are solutions to $\det{B(z)}=0$, which in this case is a polynomial with degree at most $2n$. In general, if $B_2$ is not invertible, there will be strictly less than $2n$ eigenvalues and there will be \textit{eigenvectors at infinity}: these are defined as eigenvectors of the matrix polynomial $z^2B(z^{-1})$ at $z=0$.

Secondly, since there may be more than $n$ eigenvectors, the QEP admits linearly \textit{dependent} eigenvectors. Another different feature is that one eigenvector can have two different eigenvalues.

Another analogy is when there is not enough eigenvectors: in the ordinary problem, $A$ might not be diagonalizable, but we can still write it in the Jordan normal form, by finding its generalized eigenvectors. The same situation applies to the QEP; we will see examples of these features shortly.

Returning to the discussion of bound states, the restriction of evanescent bound states to the graph are then defined as eigenvectors of $\gamma(z)$ with eigenvalue $z\in (-1,1)$, that are orthogonal to $\mathcal{C}$, with the other amplitudes defined by \eqref{otheramp}.

On a side-note, on the orthogonal complement to $\mathcal{C}$, there may be ``almost'' bound states with eigenvalue $z=\pm 1$, called half-bound states in the literature. They are \textit{not} normalizable and do not correspond to any scattering state for any momentum. They will be a part of the structure of $\gamma(z)$, but will not play the role of a bound state in further discussions.

Let us go back to the confined bound states; $\ket{\psi_c}$ satisfies \eqref{QEP0} for any $z\in\mathbb{C}$ such that $z+z^{-1}=\lambda_c$. In other words, $\gamma(z)$ has two eigenvalues\footnote{If $z=\pm1$, there is actually one eigenvalue with algebraic multiplicity $2$ associated with only one eigenvalue; this will be discussed in the next section.}, inverse to each other, for each confined bound state. Note that $\lambda_c$ can be any real number, which means those eigenvalues are real or are on the unit circle $S^1$. For $z_c\in S^1$, those bound states fall exactly in the continuum band of scattering states, with energies between $-2$ and $2$, and that is what renders the matrix $Q(z)$ not invertible for $z=z_c$. We have characterized the eigenvalues of $\gamma(z)$ related to bound states: denoting $(n_{ev})$ $n_c$ the dimension of the subspace of (un)confined bound states, there are $n_{ev}$ eigenvalues $z\in(-1,1)$, there are $2n_c$ eigenvalues in $\mathbb{R}\cup S^1$, counting with multiplicities, in addition to possible $\pm1$ eigenvalues associated with $n_h$ half-bound states. Let us choose an ordering for the eigenvalues and the bound states; given a list of eigenvalues $\{z_i\}_{i\in\mathcal{I}}$ of $\gamma(z)$ with corresponding eigenvectors $\ket{x_i}$ satisfying the properties described above, we define define four index sets $\mathcal{I}_{c},\mathcal{I}_{c}^{\pm1},\mathcal{I}_{ev},\mathcal{I}_{hb}$ associated respectively with:

Confined bound states with $z_i\neq\pm1$:
\begin{equation}
\begin{aligned}
\mathcal{I}_{c}&=\{i\in\mathcal{I}\,|\, z_i\in\mathbb{R}\cup S^1 \setminus \{\pm1\},  \bra{x_i}(\mathds{1}-P_M)\ket{x_i}=0\};
\end{aligned}
\end{equation}

Confined bound states with $z_i=\pm1$:
\begin{equation}
\begin{aligned}
\mathcal{I}_{c}^{\pm1}&=\{i\in\mathcal{I}\,|\,z_i=\pm1, \bra{x_i}(\mathds{1}-P_M)\ket{x_i}=0\};
\end{aligned}
\end{equation}

Evanescent bound states:
\begin{equation}
\begin{aligned}
\mathcal{I}_{ev}&=\{i\in\mathcal{I}\,|\,z_i\in(-1,1), \bra{x_i}(\mathds{1}-P_M)\ket{x_i}\neq0\};
\end{aligned}
\end{equation}

Half-bound states:
\begin{equation}
\begin{aligned}
\mathcal{I}_{hb}&=\{i\in\mathcal{I}\,|\,z_i=\pm1, \bra{x_i}(\mathds{1}-P_M)\ket{x_i}\neq0\}.
\end{aligned}
\end{equation}

We say that $\{\ket{x_i}\,|i\in\mathcal{I}_c\cup\mathcal{I}_{c}^{\pm1}\cup\mathcal{I}_{ev}\}$ are \textit{physical} states. Note that for each confined bound state $\ket{x_i}$ with $i\in\mathcal{I}_{c}$, there is a different index $j\in\mathcal{I}_{c}$ associated with the eigenvalue $z_j=z_i^{-1}$ of the same state.

In the next section we analyze the structure of the operator $\gamma(z)$ in more detail.

\section{Resolvent form of the gamma matrix}
\label{apx:C}
When we study the two-particle scattering on general graphs, we will need more results relating scattering states and $\gamma(z)$, defined in Eq.\ \eqref{Gamma}; in particular, we will be interested in the inverse of the matrix $\gamma (z)$, known as the resolvent form of the matrix polynomial. First we note that $\gamma$ acts separately on the subspace of confined states and its orthogonal complement:
\begin{equation}
\begin{aligned}
\gamma(z)\ket{x_i}=(-z^2P_M+zH_G-\mathds{1})\ket{x_i}=(-z^2+(z_i+z_i^{-1}) z+1)\ket{x_i},
\end{aligned}
\end{equation}
for $i\in\mathcal{I}_c\cup\mathcal{I}_c^{\pm1}$.

In other words, $\ket{x_i}$ are eigenvectors in the usual sense of $\gamma(z)$. Defining the projector on the confined subspace $P_\mathcal{C}:=\sum_{i\in\mathcal{I}_c\cup\mathcal{I}_c^{\pm1}} \ket{x_i}\bra{x_i}$, we have a direct sum structure:
\begin{equation}
\begin{aligned}
\gamma(z)&=\sum_{i\in\mathcal{I}_c\cup\mathcal{I}_c^{\pm1}} (-z^2+(z_i+z_i^{-1}) z+1)\ket{x_i}\bra{x_i}\oplus(\mathds{1}-P_\mathcal{C})\gamma(z)(\mathds{1}-P_\mathcal{C})
\\
\gamma(z)^{-1}&=-\sum_{i\in\mathcal{I}_c\cup\mathcal{I}_c^{\pm1}} \frac{\ket{x_i}\bra{x_i}}{(z-z_i)(z-z_i^{-1})}\oplus\gamma_{nc}(z)^{-1},
\label{gammainvdirectsum}
\end{aligned}
\end{equation}
where $\gamma_{nc}(z)=(\mathds{1}-P_\mathcal{C})\gamma(z)(\mathds{1}-P_\mathcal{C})$ is the gamma matrix restricted to the orthogonal complement of $\mathcal{C}$; in the following discussion, every matrix equation should be considered on the orthogonal complement of $\mathcal{C}$. The inverse of $\gamma_{nc}$ can be obtained as follows (ref. \cite{matrixpolynomials} theorem 7.7):
\begin{equation}
\begin{aligned}
\gamma_{nc}(z)^{-1}&=\begin{bmatrix}
    X_F & X_\infty
\end{bmatrix}\begin{bmatrix}
    z\mathds{1}-J_F & 0 \\
    0 & zJ_{\infty}-\mathds{1}
\end{bmatrix}^{-1}\begin{bmatrix}
    Y_F \\
    Y_\infty
\end{bmatrix},
\end{aligned}
\label{gammainv0}
\end{equation}
where the columns $\ket{x_i}$ of $X_F$ are the generalized eigenvectors associated with finite eigenvalues $z_i$ of $\gamma_{nc}(z)$, satisfying $\gamma_{nc}(z_i)\ket{x_i}=0$, and $J_F$ is its corresponding Jordan matrix\footnote{We shall see shortly that the ordering we chose in the last section can be included as a subset of this list.}. This list may include eigenvectors of $\gamma_{nc}(z)$ that are not associated with any bound state; we will refer to them as non-physical eigenvectors. Furthermore,  $J_\infty$ is the zero eigenvalue Jordan matrix for the eigenvectors at infinity, namely, the columns of $X_\infty$ are the generalized eigenvectors at zero of $M(z):=z^2 \gamma_{nc}(z^{-1})$. The matrices $Y_F,Y_\infty$ are given by:
\begin{equation}
\begin{aligned}
\begin{bmatrix}
    Y_F \\
    Y_\infty
\end{bmatrix}&=
\begin{bmatrix}
    \mathds{1} & 0 \\
    0 & J_\infty
\end{bmatrix}
\begin{bmatrix}
    X_F & X_\infty \\
    -P_m X_F J_F & X_\infty J_\infty-HX_\infty
\end{bmatrix}^{-1}
\begin{bmatrix}
    0 \\
    \mathds{1}
\end{bmatrix}\\
&=:\begin{bmatrix}\bra{y_1} \\ \cdots \\ \bra{y_{2N}}\end{bmatrix}.
\end{aligned}
\end{equation}

Let us characterize the generalized eigenvectors at infinity. We will compute them explicitly when the graph has exactly two connecting rails on vertices $1$ and $2$. Note that $M(z)=-z^2\mathds{1}+zH_G-P_M$; we have
\begin{equation}
\begin{aligned}
M(0)&=\begin{bmatrix}
    0_{2\times 2} & 0 \\
    0 & -\mathds{1}_{M\times M}
\end{bmatrix}=-P_M,
\end{aligned}
\end{equation}
so, independently of $H_G$, we have two eigenvectors at infinity, corresponding to the subspace related to the first two vertices. A $k$-chain of generalized eigenvectors for the eigenvalue $z_0$ is defined as a simultaneous solution to the $k$ equations   
\begin{equation}
\begin{aligned}
\sum_{p=0}^i \dfrac{1}{p!}M^{(p)}(z_0)\ket{\tilde{x}_{i-p}}=0, \, i=0,\cdots,k-1,
\end{aligned}
\end{equation}
where $M^{(p)}$ denotes the $p$-th derivative of $M$ with respect to $z$. For $i=0$, we have the usual eigenvector equation $M(0)\ket{\tilde{x}_0}=0$. Let us begin with one such eigenvector:
\begin{equation}
\begin{aligned}
\ket{\tilde{x}_0}=\begin{bmatrix}\alpha_0 \\ \beta_0 \\ \vdots \\ 0\end{bmatrix}=\alpha_0\ket{1}+\beta_0\ket{2}
\end{aligned}
\end{equation}
For $i\ge 1$, we have:
\begin{equation}
\begin{aligned}
H_G\ket{\tilde{x}_0}&-P_M\ket{\tilde{x}_1}=0,\\
-\ket{\tilde{x}_0}+H_G\ket{\tilde{x}_1}&-P_M\ket{\tilde{x}_2}=0,\\
-\ket{\tilde{x}_1}+H_G\ket{\tilde{x}_2}&-P_M\ket{\tilde{x}_3}=0,\\
&\,\,\,\vdots\\
-\ket{\tilde{x}_{k-3}}+H_G\ket{\tilde{x}_{k-2}}&-P_M\ket{\tilde{x}_{k-1}}=0.
\label{geneig}
\end{aligned}
\end{equation}

Applying $\mathds{1}-P_M$ to the first equation, we have the condition for the existence of $\ket{\tilde{x}_1}$:
\begin{equation}
\begin{aligned}
0=(\mathds{1}-P_M)H\ket{\tilde{x}_0},
\end{aligned}
\end{equation}
which is satisfied if, and only if there are no edges between vertices $1$ and $2$. In this case, we have
\begin{equation}
\begin{aligned}
\ket{\tilde{x}_1}=\alpha_1\ket{1}+\beta_1\ket{2}+\alpha_0 H_G\ket{1}+\beta_0 H_G\ket{2},
\end{aligned}
\end{equation}
for some $\alpha_1,\beta_1\in\mathbb{C}$. Applying $\mathds{1}-P_M$ to the second equation in \eqref{geneig}, we have a condition for the existence of $\ket{\tilde{x}_2}$:
\begin{equation}
\begin{aligned}
-(\mathds{1}-P_M)\ket{\tilde{x}_0}+(\mathds{1}-P_M)H_G\ket{\tilde{x}_1}&=0\\
-\alpha_0\ket{1}-\beta_0\ket{2}+\alpha_0 d_1\ket{1}+\alpha_0 p_{12}\ket{1}+\beta_0 d_2\ket{2}+\beta_0 p_{12}\ket{2}&=0
\end{aligned}
\end{equation}
\begin{equation}
\begin{aligned}
\Longrightarrow \begin{cases}(d_1-1)\alpha_0+p_{12}\beta_0=0\\
p_{12}\alpha_0+(d_2-1)\beta_0=0 \end{cases},
\end{aligned}
\end{equation}
where $d_1$ and $d_2$ are the inner degrees of vertices $1$ and $2$ respectively, i.e., without accounting for edges outside of $G$, and $p_{12}$ is the number of paths of length $2$ between vertices $1$ and $2$. This system has no non-trivial solution if and only if 
\begin{equation}
\begin{aligned}
(d_1-1)(d_2-1)-p^2_{12}\neq0.
\label{chaincondition}
\end{aligned}
\end{equation}
We will restrict ourselves to graphs that do not have an edge between vertices $1$ and $2$ and satisfy this condition, with the only exception being the linear graph, which has already been considered in \cite{luna}. Consequently both Jordan chains at infinity will have length exactly $2$ and the Jordan matrix corresponding to those vectors is
\begin{equation}
\begin{aligned}
J_\infty&=\begin{bmatrix}
    0 & 1 & 0 & 0 \\
    0 & 0 & 0 & 0 \\
    0 & 0 & 0 & 1 \\
    0 & 0 & 0 & 0
\end{bmatrix}.
\end{aligned}
\end{equation}
We have $J_\infty^2=0$, which implies $(zJ_\infty-\mathds{1})^{-1}=-(\mathds{1}+zJ_\infty)$. It follows from the lemmas in \cite{levinsonthm2} that all Jordan chains associated with evanescent and half bound states have length one\footnote{This is not the case for confined bound states; this is the reason we treated them separately.}; in this case we can separate even further and simplify:
\begin{equation}
\begin{aligned}
\gamma_{nc}(z)^{-1}&=\begin{bmatrix}
    X_{\leq 1} & X_{>1} & X_\infty
\end{bmatrix}\begin{bmatrix}
    diag\dfrac{1}{z-z_i} & 0 & 0 \\ 0 & (z\mathds{1}-J)^{-1} & 0 \\
    0 & 0 & -(\mathds{1}+zJ_{\infty})
\end{bmatrix}\begin{bmatrix}
    Y_{\leq 1} \\
    Y_{>1} \\
    Y_\infty
\end{bmatrix}\\
&=\sum_{i\in\mathcal{I}_{ev}\cup\mathcal{I}_{hb}}\dfrac{1}{z-z_i}\ket{x_i}\bra{y_i}+X_{>1}(z\mathds{1}-J)^{-1}Y_{>1} - X_{\infty}Y_{\infty},
\end{aligned}
\label{gammainv1}
\end{equation}
where $z_i$ and $\ket{x_i}$ are chosen as in the last section, the columns of $X_{\leq1}$ are the vectors $\ket{x_i}$ related to evanescent and half bound states and the columns of $X_{>1}$ are the generalized eigenvectors with eigenvalues $|\lambda|>1$, excluding possible confined states, with the corresponding matrices $Y_{\leq1}$ and $Y_{>1}$ respectively. 
Going back to equation \eqref{gammainvdirectsum}, we have the full expression
\begin{equation}
\begin{aligned}
\gamma(z)^{-1}=&-\sum_{i\in\mathcal{I}_c\cup\mathcal{I}_c^{\pm1}} \frac{\ket{x_i}\bra{x_i}}{(z-z_i)(z-z_i^{-1})}\\
&+ \sum_{i\in\mathcal{I}_{ev}\cup\mathcal{I}_{hb}}\dfrac{1}{z-z_i}\ket{x_i}\bra{y_i}+X_{>1}(z\mathds{1}-J)^{-1}Y_{>1} - X_{\infty}Y_{\infty}.
\end{aligned}
\label{invgamma00}
\end{equation}

Let $\Lambda$ be the multiset --- as each eigenvalue may appear more than once --- of non-physical eigenvalues, $n_\lambda$ the partial multiplicities --- each eigenvalue may have different length Jordan chains --- of the eigenvalue $\lambda$, i.e, sizes of the respective Jordan block $J_\lambda=\lambda\mathds{1}+N_\lambda$ and $X_\lambda$ the respective Jordan chain with the matching $Y_\lambda$. Then we can write
\begin{equation}
\begin{aligned}
X_{>1}(z\mathds{1}-J)^{-1}Y_{>1}&=\sum_{\lambda\in\Lambda } X_\lambda (z\mathds{1}-J_\lambda)^{-1}Y_\lambda\\
&=\sum_{\lambda\in\Lambda}\sum_{ n=1}^{n_\lambda} \frac{1}{(z-\lambda)^n}X_\lambda N_\lambda^{n-1} Y_\lambda
\end{aligned}
\end{equation}

Since $\gamma(0)^{-1}=-\mathds{1}$, the term involving eigenvectors at infinity can be written in terms of the other eigenvectors:
\begin{equation}
\begin{aligned}
-X_{\infty}Y_{\infty}=-\mathds{1}+\sum_{i\in\mathcal{I}_c\cup\mathcal{I}_c^{\pm1}} \ket{x_i}\bra{x_i}-\sum_{i\in\mathcal{I}_{ev}\cup\mathcal{I}_{hb}}\frac{1}{-z_i}\ket{x_i}\bra{y_i}-\sum \frac{1}{(-\lambda)^n}X_\lambda N_\lambda^{n-1} Y_\lambda,
\end{aligned}
\end{equation}
and we can rewrite equation \eqref{invgamma00} as:
\begin{equation}
\begin{aligned}
\gamma(z)^{-1}=&-\mathds{1}+\sum_{i\in\mathcal{I}_c\cup\mathcal{I}_c^{\pm1}} \left(1-\frac{1}{(z-z_i)(z-z_i^{-1})}\right)\ket{x_i}\bra{x_i}+\sum_{i\in\mathcal{I}_{ev}\cup\mathcal{I}_{hb}}\dfrac{z}{z_i(z-z_i)}\ket{x_i}\bra{y_i}\\
&+\sum_{\lambda\in\Lambda}\sum_{ n=1}^{n_\lambda} \frac{\lambda^n-(\lambda-z)^n}{\lambda^n(z-\lambda)^n}X_\lambda N_\lambda^{n-1} Y_\lambda.
\end{aligned}
\label{gammainv3}
\end{equation}

The term associated with the confined states can be further simplified, depending on whether $z_i=\pm1$ or $z_i\neq\pm1$. For each $z_i=\pm1$, we have
\begin{equation}
\begin{aligned}
\left(1-\frac{1}{(z-z_i)^2}\right)\ket{x_i}\bra{x_i}=\frac{z(z-2z_i)}{(z-z_i)^2}\ket{x_i}\bra{x_i};
\end{aligned}
\end{equation}
this means that each confined state with $z_i=\pm1$ is associated with \textit{one} eigenvalue with algebraic multiplicity $2$; and for $z_i\neq\pm1$ we have
\begin{equation}
\begin{aligned}
\left(1-\frac{1}{(z-z_i)(z-z_i^{-1})}\right)\ket{x_i}\bra{x_i}=-\frac{z}{z_i(z-z_i)}\frac{\ket{x_i}\bra{x_i}}{z_i-z_i^{-1}}-\frac{z}{z_i^{-1}(z-z_i^{-1})}\frac{\ket{x_i}\bra{x_i}}{z_i^{-1}-z_i};
\end{aligned}
\end{equation}
note that this is very similar to the terms for evanescent bound states in Eq.\ \eqref{gammainv3}, one for each of the two eigenvalues associated with $\ket{x_i}$. Defining $\bra{y_i}:=\frac{\bra{x_i}}{z_i-z_i^{-1}}$ for each $i\in\mathcal{I}_c$, we can write a more uniform expression:
\begin{equation}
\begin{aligned}
\gamma(z)^{-1}=&-\mathds{1}+\sum_{i\in\mathcal{I}_c^{\pm1}} \frac{z(z-2z_i)}{(z-z_i)^2}\ket{x_i}\bra{x_i}+\sum_{i\in\mathcal{I}_{ev}\cup\mathcal{I}_{hb}\cup\mathcal{I}_c}\dfrac{z}{z_i(z-z_i)}\ket{x_i}\bra{y_i}\\
&+\sum_{\lambda\in\Lambda}\sum_{ n=1}^{n_\lambda} \frac{\lambda^n-(\lambda-z)^n}{\lambda^n(z-\lambda)^n}X_\lambda N_\lambda^{n-1} Y_\lambda.
\end{aligned}
\label{invgammafinal}
\end{equation}

This expression is the resolvent form for $\gamma(z)^{-1}$; as an application, let us show that $\ket{y_i}$ is proportional to $\ket{x_i}$, as is the case for the confined bound states. The authors of \cite{levinsonthm2} show that, for $x\in\mathbb{R}$, there exist functions $e_i$ and an orthonormal basis $\ket{v_i(x)}$ satisfying
\begin{equation}
\begin{aligned}
\gamma_{nc}(x)^{-1}=\sum_{j=1}^{M+N-n_c}\frac{\ket{v_j(x)}\bra{v_j(x)}}{e_j(x)},
\end{aligned}
\end{equation}
where $\ket{v_j(x)}$ are eigenvectors of $\gamma_{nc}(x)$ in the usual sense: $\gamma(x)\ket{v_j(x)}=e_j(x)\ket{v_j(x)}$. In addition to this expression, every zero $x_0$ of $e_j(x)$ corresponds to either a half-bound state or an evanescent bound state. In the latter case, it holds that
\begin{equation}
\begin{aligned}
\text{Res}_{x=x_0}\frac{\ket{v_j(x)}\bra{v_j(x)}}{e_j(x)}=-\frac{\ket{\phi_b}\bra{\phi_b}}{x_0-x_0^{-1}},
\end{aligned}
\end{equation}
for some evanescent bound state $\ket{\phi_b}$.


Let $x_0\in(-1,1)$. The discussion above implies that
\begin{equation}
\begin{aligned}
\text{Res}_{x=x_0}\gamma_{nc}(x)^{-1}=-\sum_{i\in\mathcal{I}_{ev}:z_i=x_0}\frac{\ket{x_i}\bra{x_i}}{x_0-x_0^{-1}}.
\end{aligned}
\label{andrewresidue}
\end{equation}
On the other hand, using the resolvent form, we have
\begin{equation}
\begin{aligned}
\text{Res}_{x=x_0}\gamma_{nc}(z)^{-1}=\sum_{i\in\mathcal{I}_{ev}:z_i=x_0}\ket{x_i}\bra{y_i}.
\end{aligned}
\end{equation}
Furthermore, even though $\ket{x_i}$ is only defined on $G$, since the bound states are orthogonal to each other in the whole space, we have
\begin{equation}
\begin{aligned}
(z_i-z_i^{-1})\ket{x_i}\bra{y_i}=-\ket{x_i}\bra{x_i},
\end{aligned}
\end{equation}
implying that $\bra{y_i}=-\dfrac{\bra{x_i}}{z_i-z_i^{-1}}$. Now we have the possibility to write an expression for \eqref{invgammafinal} that does not explicitly depend on the vectors $\bra{y_i}$, which themselves are related to the columns of $X_F$ and $X_{\infty}$; in other words, the residues of $\gamma(z)^{-1}$ on physical eigenvalues depend explicitly only on the physical states $\ket{x_i}$.

\section{Graph examples}

\label{apx:D}
In this Section, we define some graphs which we use to exemplify new phenomena in sections \ref{sec:4B} and \ref{sec:4C}. Since the motivation is developing gadgets to transform the state of freely moving particles on rails, we will consider two-particle interactions with finite support, i.e., that act on a finite number of position states $\ket{uv}$; without loss of generality, we consider that the particles will only interact on the inner vertices of the graph, $G^0$. This setting encompasses all possible interactions with finite support. 

The first example we will study is the linear graph with $M$ inner vertices, \texttt{Line(M)}, shown in Fig.\ \ref{fig:Line5}:

\begin{figure}[h]
    \centering    \includegraphics[width=0.7\textwidth]{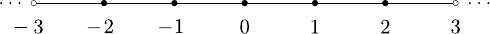}
    \caption[\texttt{Line(5)}]{Depiction of \texttt{Line(5)}; the particles only interact on the central five vertices.}
    \label{fig:Line5}
\end{figure}

Since the linear graph with two rails attached is the same as a infinite line lattice, we label the vertices as integers and consider a symmetric region of interaction. This graph is the only exception to the discussion of the last sections; it does not have any bound states, and both Jordan chains at infinity have length greater than $2$. This will not be a problem, since we will not need to use the resolvent form of $\gamma(z)^{-1}$ to find the S-matrix for the line.

\begin{figure}[h]
    \centering    \includegraphics[width=0.5\textwidth]{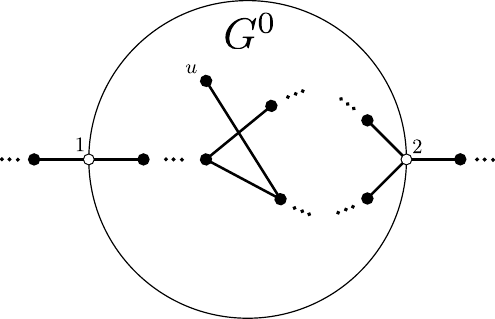}
    \caption[Graph with two rails]{A generic graph with two rails; $u$ is a typical inner vertex in $G^0$, vertices $1$ and $2$ have inner degree $d_1=1$ and $d_2=2$ respectively. We do not consider graphs with edges between $1$ and $2$. The white dots represent vertices where we attach a rail.}
    \label{fig:generic graph 2 rails}
\end{figure}

Every other graph example is similar to the one shown in Fig.\ \ref{fig:generic graph 2 rails}; we consider graphs with only two rails, connecting to different vertices $1$ and $2$, with no edge between those vertices.

\begin{figure}[htpb]
    \centering    \includegraphics[width=\textwidth]{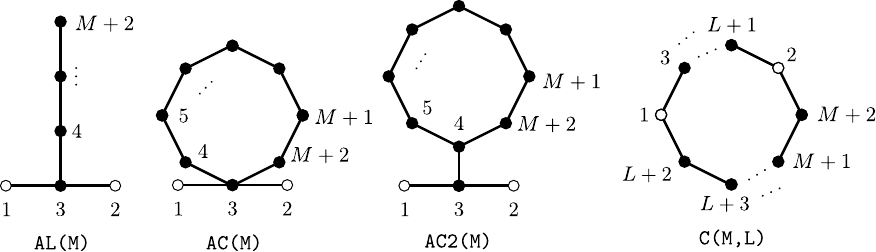}
    \caption[Appended subgraphs]{From left to right: \texttt{AL}(M) is the appended linear graph, \texttt{AC}(M) is the appended cycle graph, \texttt{AC2}(M) is the second family for appended cycle graphs and \texttt{C(M,L)} is a cycle with $M+2$ vertices in total, with rails positioned such that there are $L$ edges in the shortest path between $1$ and $2$..}
    \label{fig:append examples}
\end{figure}

Let us consider three families of graphs which only append a subgraph to the line, shown in Fig.\ \ref{fig:append examples}: the first two are the appended linear graph \texttt{AL(M)} and the appended cycle graph \texttt{AC(M)}. The third one is similar to the appended cycle, with an extra edge between vertex $3$ and the proper cycle. We name it \texttt{AC2(M)}; note that it has a cycle of size $M-1$. All those graphs have no edge between $1$ and $2$ and satisfy $d_1=d_2=1$ and $p_{12}=1$, implying $(d_1-1)(d_2-1)\neq p^2_{12}$. This means that both Jordan chains at infinity have length $2$, and the discussion from last section applies to these graphs. The last we consider are the cycle graphs; we parametrized the cycle by the number of interaction vertices $M$ and the length $L$ of the shortest path between vertices $1$ and $2$. To avoid double counting, $L$ must be at most $\lfloor{\frac{M+2}{2}}\rfloor$. For values $3\leq L\leq\lfloor{\frac{M+2}{2}}\rfloor$ the graph always satisfies $(d_1-1)(d_2-1)\neq p^2_{12}$, since $d_1=d_2=2$ and $p_{12}=0$; the graph \texttt{C(2,2)} also satisfies this condition.
\begin{table}[htpb]
    \centering
    \begin{tabular}{ |p{3.5cm}||p{1cm}|p{1cm}|p{1cm}|  } 
 \hline
 Graph& $n_{ev}$ & $n_c$ & $n_h$\\
 \hline
 \texttt{Line(M)} $M\geq0$ & $0$ & $0$ &  $2$\\
 \texttt{AL(M)} $M\geq2$  &  $2$    & $0$ &   $0$\\
 \texttt{AC(M)} $M\geq3$  &  $2$ & $\lceil\frac{M}{2}\rceil$ & $0$\\
 \texttt{AC2(M+1)} $M\geq3$   & $2$ & $\lceil\frac{M}{2}\rceil $ &  $0$\\
 \texttt{C(2,2)} & $2$ & $1$ &  $0$ \\
 \texttt{C(4,3)} & $2$ & $2$ &  $2$ \\
 \texttt{C(2L-2,L)} $L\geq3$   & $4$ & $L-1$ &  0\\
 \hline
\end{tabular}
    \caption{Number of evanescent $(n_{ev})$, confined $(n_{c})$ and half bound $(n_{hb})$ states. The formulas depending on $M$ has been checked up to graphs with $20$ vertices.}
    \label{tab:bound}
\end{table}

In Table \ref{tab:bound} we count the number of bound states of such graphs; the values have been checked for graphs with up to $20$ vertices.

\section{Two-particle S-matrix computation}
\label{apx:E}

We want to compute the two-particle S-matrix:
\begin{equation}
\begin{aligned}
&S_{\zeta_1\zeta_2;\xi_1\xi_2}=\braket{\zeta_1\zeta_2|\xi_1\xi_2}-2\pi i\delta(E_{\zeta_1\zeta_2}-E_{\xi_1\xi_2})\bra{\zeta_1\zeta_2}V\ket{{\xi_1\xi_2}^+},
\end{aligned}
\end{equation}

For simplicity, let us assume that the particles can only interact when they are on the same site, with the same interaction energy, namely, let $V=U\sum_{u\in G_0}\ket{uu}\bra{uu}$ for some $U\in\mathbb{R}$. Then, to compute the S-matrix, we only require the amplitudes when particles are on the same site, namely, terms of the type $\braket{uu|\xi_1\xi_2}$. We can write in general, from the definition of $V$:
\begin{equation}
\begin{aligned}
\bra{\zeta_1\zeta_2}V\ket{{\xi_1\xi_2}^+}=bU\sum_{u\in G^0} \braket{\zeta_1|u}\braket{\zeta_2|u}\braket{uu|{\xi_1\xi_2}^+},
\label{kVp}
\end{aligned}
\end{equation}
where $b = 1$ for distinguishable particles, $b = \sqrt{2}$ for bosons and $b = 0$ for fermions.

The terms $\braket{uu|{\xi_1\xi_2}^+}$ can be found by taking the inner product of the corresponding Lippmann-Schwinger equation with $\bra{uu}$ for $u\in G^0$:
\begin{equation}
\begin{aligned}
\braket{uu|{\xi_1\xi_2}^+}&=b\braket{u|\xi_1}\braket{u|\xi_2}+U\sum_{v\in G^0}J_{uv} \braket{vv|{\xi_1\xi_2}^+},
\label{lipp}
\end{aligned}
\end{equation}
where
\begin{equation}
\begin{aligned}
J_{uv}=\int_{-\pi}^0\int_{-\pi}^0&\sum_{m_1, m_2\in \partial G}\frac{\braket{u|{k_1^{m_1}}^+}\braket{{k_1^{m_1}}^+|v}\braket{u|{k_2^{m_2}}^+}\braket{{k_2^{m_2}}^+|v}}{E-2\cos{k_1}-2\cos{k_2}+ i\epsilon}dk_1 dk_2  \\
 +2\int_{-\pi}^0&\sum_{\substack{m\in\partial G\\ i\in\mathcal{I}_{bound}}}\frac{\braket{u|{k^{m}}^+}\braket{{k^{m}}^+|v}\braket{u|x_i}\braket{x_i|v}}{E-2\cos{k}-E_i+ i\epsilon}dk\\
+&\sum_{i,j\in\mathcal{I}_{bound}}\frac{\braket{u|x_{i}}\braket{x_{i}|v}\braket{u|x_{j}}\braket{x_{j}|v}}{E-E_i-E_j+i\epsilon},
\label{Juv}
\end{aligned}
\end{equation}
where $\mathcal{I}_{bound}$ corresponds to all bound states, both confined and evanescent, without double counting of the confined states, $E_i:=z_i+z_i^{-1}$ is the energy of the bound state $\ket{x_i}$ and we have used the resolution of the identity $\mathds{1}=\int_{-\pi}^0 \sum_{m\in \partial G}\ket{{k^m}^+}\bra{{k^m}^+}dk + \sum_{i\in\mathcal{I}_{bound}} \ket{x_i}\bra{x_i}$, valid on $G$.



\begin{lemma}
Let $u,v\in G^0$. Then the following equations hold (for the definitions of 
$\Psi$ and $\gamma$ refer to equations \eqref{Psi} and \eqref{Gamma}):
\begin{equation}
\begin{aligned}
\sum_{m\in\partial G} \braket{u|{k^{m}}^+}\braket{{k^{m}}^+|v}=\frac{1}{2\pi}\bra{u}\Psi\Psi^\dagger\ket{v}=\frac{1}{2\pi}\bra{u}\left[(z^2-1)\gamma(z)^{-1}+(z^{-2}-1)\gamma(z^{-1})^{-1}\right]\ket{v}.
\label{PsiPsidagger}
\end{aligned}
\end{equation}
\end{lemma}
\begin{proof}
The first equation follows directly from the definition of $\Psi$. We can verify by direct computation \cite{levinsonthm2} a more compact relation between $\gamma(z)$ and the scattering states:
\begin{equation}
\begin{aligned}
-\gamma(z)^{-1}\gamma(z^{-1})&=\begin{bmatrix}
    -\mathbb{S}(z) & 0 \\
    z^{-1}\Psi(z) & -z^{-2}
\end{bmatrix}.
\end{aligned}
\label{S&Psi}
\end{equation}

Then we can write
$$\bra{u}\left(\Psi\Psi^\dagger+\mathds{1}\right)\ket{v}=\bra{u}\gamma(z)^{-1}\gamma(z^{-1})\gamma(z)\gamma(z^{-1})^{-1}\ket{v}.$$
Note that $\gamma(z^{-1})=z^{-2}\gamma(z)+\left(z^{-2}-1\right)P_N$, where $P_N=\mathds{1}-P_M$, so
$$\bra{u}\Psi\Psi^\dagger\ket{v}=(z^2-1)(z^{-2}-1)\bra{u}\gamma(z)^{-1}P_N\gamma(z^{-1})^{-1}\ket{v}$$
Also note that, from the definition $\gamma(z)=-z^2 (\mathds{1}-P_N)+ z H_G -\mathds{1}$, we have $\gamma(z)^{-1}P_N=z^{-2}-z^{-1}\gamma(z)^{-1}(H-z-z^{-1})$. Using this expression and its hermitian conjugate, we can write the previous equation in two different ways:
\begin{equation}
\begin{aligned}
\bra{u}\Psi\Psi^\dagger\ket{v}&=(z^2-1)(z^{-2}-1)\bra{u}\left(z^{-2}-z^{-1}\gamma(z)^{-1}(H-z-z^{-1})\right)\gamma(z^{-1})^{-1}\ket{v}:=A\\
&=(z^2-1)(z^{-2}-1)\bra{u}\gamma(z)^{-1}\left(z^2-z(H-z-z^{-1})\gamma(z^{-1})^{-1}\right)\ket{v}:=B
\end{aligned}
\end{equation}
Since $A=B$, we can write
\begin{equation}
\begin{aligned}
(z^2-1)\bra{u}\Psi\Psi^\dagger\ket{v}=z^2A-B=(z^2-1)(z^{-2}-1)\bra{u}(\gamma(z^{-1})^{-1}-z^2\gamma(z)^{-1})\ket{v},
\end{aligned}
\end{equation}
from which follows directly Eq.\ \eqref{PsiPsidagger}.
\end{proof}
Define $P:=\mathds{1}-P_\mathcal{C}$ the projector on the orthogonal complement of the confined subspace. Note that $P_\mathcal{C}\Psi=0$. It follows from the fact that the confined bound states are orthogonal to the restriction of the scattering states on the graph.

So, we can compute:
\begin{equation}
\begin{aligned}
\frac{1}{2}\oint \frac{\bra{u}\Psi\Psi^\dagger\ket{v}}{E-E'+i\epsilon-z-z^{-1}}\frac{dz}{2\pi iz}=&\frac{1}{2}\oint \frac{\bra{u}P\Psi\Psi^\dagger\ket{v}}{E-E'+i\epsilon-z-z^{-1}}\frac{dz}{2\pi iz}\\
=&\oint \frac{\bra{u}P(z^2-1)\gamma(z)^{-1}\ket{v}}{E-E'+i\epsilon-z-z^{-1}}\frac{dz}{2\pi i z}\\
=-\omega^-_{E'}\bra{u}&P\gamma(\omega^-_{E'})^{-1}\ket{v}+\sum_{i\in\mathcal{I}:z_i\in(-1,1)}z_i \frac{\braket{u|x_i}\braket{x_i|v}}{(z_i-\omega^-_{E'})(z_i-\omega^+_{E'})}\\
=-\omega^-_{E'}\bra{u}&P\gamma(\omega^-_{E'})^{-1}\ket{v}-\sum_{i\in\mathcal{I}_{ev}} \frac{\braket{u|x_i}\braket{x_i|v}}{E-E'-E_i+i\epsilon}
\end{aligned}
\label{intpsipsi}
\end{equation}
where $E'$ stands for either $2\cos{k_2}$ in the double integral or $E_i$ in the single integral , $$\omega^{\pm}_{E'}=\frac{1}{2}\left(E-E'+i\epsilon\pm\sqrt{E-E'+i\epsilon-2}\sqrt{E-E'+i\epsilon+2}\right),$$ and the branch is chosen such that $\omega^-_{E'}$ has smaller absolute value; this choice is exactly the principal branch for the square roots. Note that $\omega^+_{E'}=(\omega^{-}_{E'})^{-1}$.

Let us apply Eq.\ \eqref{intpsipsi} to Eq.\ \eqref{Juv}: changing the variable to $z=e^{ik}$ in the single integral, we get
\begin{equation}
\begin{aligned}
-2\sum_{i\in\mathcal{I}_{bound}}\omega^-_{E_i}\bra{u}P\gamma (\omega^-_{E_i})^{-1}\ket{v}\braket{u|x_i}\braket{x_i|v}-2\sum_{\substack{i\in\mathcal{I}_{bound}\\j\in\mathcal{I}_{ev}}}\frac{\braket{u|x_i}\braket{x_i|v}\braket{u|x_j}\braket{x_j|v}}{E-E_i-E_j+ i\epsilon};
\end{aligned}
\end{equation}
analogously, applying Eq.\ \eqref{intpsipsi} twice to the double integral, we get
\begin{equation}
\begin{aligned}
&-\oint \omega^-_{z+z^{-1}}\bra{u}P\gamma(\omega^-_{z+z^{-1}})^{-1}\ket{v}\bra{u}(z^2-1)P\gamma(z)^{-1}\ket{v}\frac{dz}{2\pi i z}\\
&+\sum_{i\in\mathcal{I}_{ev}}\omega^-_{E_i}\bra{u}P\gamma (\omega^-_{E_i})^{-1}\ket{v}\braket{u|x_i}\braket{x_i|v}+\sum_{i,j\in\mathcal{I}_{ev}}\frac{\braket{u|x_i}\braket{x_i|v}\braket{u|x_j}\braket{x_j|v}}{E-E_i-E_j+ i\epsilon}
\end{aligned}
\end{equation}

Finally, putting all terms of \eqref{Juv} together, we have
\begin{equation}
\begin{aligned}
J_{uv}=&-\oint \omega^-_{z+z^{-1}}\bra{u}P\gamma(\omega^-_{z+z^{-1}})^{-1}\ket{v}\bra{u}(z^2-1)P\gamma(z)^{-1}\ket{v}\frac{dz}{2\pi iz}
\\
&-\sum_{i\in\mathcal{I}_{ev}\cup\mathcal{I}_c\cup\mathcal{I}_c^{\pm1}} \omega^-_{E_i}\bra{u}P\gamma (\omega^-_{E_i})^{-1}\ket{v}\braket{u|x_i}\braket{x_i|v}\\
&+\sum_{i,j\in\mathcal{I}_c\cup\mathcal{I}_c^{\pm1}}\frac{\braket{u|x_i}\braket{x_i|v}\braket{u|x_j}\braket{x_j|v}}{E-E_{i}-E_{j}+i\epsilon},
\label{J}
\end{aligned}
\end{equation}
thus we have reduced the work of computing $J$ to computing bound states and the gamma matrix, as in \ref{apx:B}; for some graphs, the inverse of $\gamma(z)$ may be easily computed directly, but for a systematic method, refer to Appendix \ref{apx:C}.

It is possible to show that the residues of the integrand at the poles inside the unit circle cancel out; so the integral can be computed either on top of the unit circle or on top of the branch cut of $\omega^-_{z+z^{-1}}$, independently of poles that may exist inside the circle.

Going back to the linear system of Eq.\ \eqref{lipp}, it can be written as $Tx=c$ where $x_v:=\braket{vv|{\xi_1\xi_2}^+}$, $c_u:=b\braket{u|\zeta_1}\braket{u|\zeta_2}$ and $T_{uv}:=\delta_{uv}-UJ_{uv}$.

Looking back at Eq.\ \eqref{kVp}, we can write
\begin{equation}
\begin{aligned}
\bra{\zeta_1\zeta_2}V\ket{{\xi_1\xi_2}^+}&=b U\sum_{u,v\in G^0}\braket{\zeta_1|u}\braket{\zeta_2|u}T^{-1}_{uv}\braket{v|\xi_1}\braket{v|\xi_2}\\
&=b \sum_{u,v\in G^0}\braket{\zeta_1|u}\braket{\zeta_2|u}\left(\frac{1}{U}\mathds{1}-J\right)^{-1}_{uv}\braket{v|\xi_1}\braket{v|\xi_2},
\label{E13}
\end{aligned}
\end{equation}
and finally the S-matrix is
\begin{equation}
\begin{aligned}
S_{\zeta_1\zeta_2;\xi_1\xi_2}=\braket{\zeta_1\zeta_2|\xi_1\xi_2}
-2\pi i \delta_E b^2\sum_{u,v\in G^0}\braket{\zeta_1|u}\braket{\zeta_2|u}\left(\frac{1}{U}\mathds{1}-J\right)^{-1}_{uv}\braket{v|\xi_1}\braket{v|\xi_2},
\label{E14}
\end{aligned}
\end{equation}
where $\delta_E=\delta(E_{\zeta_1\zeta_2}-E_{\xi_1\xi_2})$.

If the particles interact on specified pairs of vertices, say, for $(u,v)\in H\subseteq G_0\times G_0$ we can write the potential as $V=\displaystyle\sum_{(u_1,u_2),(v_1,v_2)\in H} v_{u_1u_2;v_1v_2}\ket{u_1u_2}\bra{v_1v_2}$; then the same calculation above remains valid, mostly with the same structure; however, all matrices are $|H|\times|H|$, instead of $M\times M$. Note that $|H|\leq M(M+1)$, so we have at most a quadratic overhead if we want to consider more general potentials. Then Eq.\eqref{E14} becomes
\begin{equation}
\begin{aligned}
S_{\zeta_1\zeta_2;\xi_1\xi_2}=\braket{\zeta_1\zeta_2|\xi_1\xi_2}
-2\pi i \delta_E\sum_{\substack{(u_1,u_2)\in H\\(v_1,v_2)\in H}}\braket{\zeta_1\zeta_2|u_1u_2}\left[V(\mathds{1}-JV)^{-1}\right]_{u_1u_2;v_1v_2}\braket{v_1v_2|\xi_1\xi_2},
\end{aligned}
\end{equation}
where the new $J$ is a $|H|\times|H|$ matrix:
\begin{equation}
\begin{aligned}
J_{u_1u_2;v_1v_2}=&-\frac{1}{2}\oint \omega^-_{z+z^{-1}}\bra{u_1}P\gamma(\omega^-_{z+z^{-1}})^{-1}\ket{v_1}\bra{u_2}(z^2-1)P\gamma(z)^{-1}\ket{v_2}\frac{dz}{2\pi iz}
\\
&-\frac{1}{2}\oint \omega^-_{z+z^{-1}}\bra{u_2}P\gamma(\omega^-_{z+z^{-1}})^{-1}\ket{v_2}\bra{u_1}(z^2-1)P\gamma(z)^{-1}\ket{v_1}\frac{dz}{2\pi iz}
\\
&-\frac{1}{2}\sum_{i\in\mathcal{I}_{ev}\cup\mathcal{I}_c\cup\mathcal{I}_c^{\pm1}} \omega^-_{E_i}\bra{u_1}P\gamma (\omega^-_{E_i})^{-1}\ket{v_1}\braket{u_2|x_i}\braket{x_i|v_2}\\
&-\frac{1}{2}\sum_{i\in\mathcal{I}_{ev}\cup\mathcal{I}_c\cup\mathcal{I}_c^{\pm1}} \omega^-_{E_i}\bra{u_2}P\gamma (\omega^-_{E_i})^{-1}\ket{v_2}\braket{u_1|x_i}\braket{x_i|v_1}\\
&+\sum_{i,j\in\mathcal{I}_c\cup\mathcal{I}_c^{\pm1}}\frac{\braket{u_1|x_i}\braket{x_i|v_1}\braket{u_2|x_j}\braket{x_j|v_2}}{E-E_{i}-E_{j}+i\epsilon}.
\end{aligned}
\end{equation}

\end{document}